\documentclass[runningheads,a4paper]{llncs}

\usepackage{amssymb}
\usepackage{graphicx}
\usepackage{cite}
\usepackage{url}
\usepackage{enumitem}

\urldef{\mailsa}\path|manoj_raut@daiict.ac.in|
\newcommand{\keywords}[1]{\par\addvspace\baselineskip
\noindent\keywordname\enspace\ignorespaces#1}

\begin{document}

\mainmatter  

\title{Computing Theory Prime Implicates in Modal Logic}

\titlerunning{Computing Theory Prime Implicates in Modal Logic}

%
%
\author{Manoj K. Raut\thanks{Author thanks the NBHM, DAE, Mumbai, India for financial support under grant reference number 2/48(16)/2014/NBHM(R.P.)/R\&D II/1392.}}
%
\authorrunning{Manoj K. Raut}

\institute{Dhirubhai Ambani Institute of Information and Communication Technology,\\
Gandhinagar, Gujarat-382007, India\\
\mailsa\\
}

%
%

\toctitle{Lecture Notes in Computer Science}
\tocauthor{Authors' Instructions}
\maketitle

\begin{abstract}
The algorithm to compute theory prime implicates, a generalization of prime implicates, in propositional logic 
has been suggested in \cite{Marquis}. In this paper we have extended that
algorithm to compute theory prime implicates of a knowledge base $X$ with respect to another knowledge base $\Box Y$ using \cite{Bienvenu}, where $Y$ is a propositional knowledge base and $X\models Y$, in modal system $\mathcal{T}$ and we have also proved its correctness. We have also proved that it is an equivalence preserving knowledge compilation and the size of theory prime implicates of $X$ with respect to $\Box Y$ is less than the size of the prime implicates of $X\cup\Box Y$. We have also extended the query answering algorithm in modal logic.
\end{abstract}

\section{Introduction}
Propositional entailment problem is a fundamental issue in artificial intelligence due to its high complexity. Determining whether a query logically follows from a given knowledge base is intractable \cite{Cook} in general as every known algorithm runs in time exponential in the size of the given knowledge base. To overcome such computational intractability, the propositional entailment problem is split into two phases such as off-line and on-line. In the off-line phase, the original knowledge base $X$ is compiled into a new knowledge base $X^{'}$ and in on-line phase queries are actually answered from the new knowledge base in time polynomial in their size. In such type of compilation most of the computational overhead shifted into the off-line phase, is amortized over large number of on-line query answering. The off-line computation is called {\em knowledge compilation}.

Several approaches of knowledge compilation in propositional logic, first order logic and modal logic has been suggested so far in literature \cite{Cadoli, Coudert1, Darwiche, Jackson1, Kean1, Kleer1, Kleer2, Ngair, Manoj, Manoj1, Reiter, Shiny, Slagle, Strzemecki, Tison}. The first kind of approach consists of an equivalence preserving knowledge compilation. In such an approach, the knowledge base $X$ is compiled into another equivalent knowledge base $\Pi(X)$, called the prime implicates of $X$ with respect to which queries are answered from $\Pi(X)$ in polynomial time. In another approach to equivalence preserving compilation in propositional logic, Marquis suggested the computation of theory prime implicates \cite{Marquis} of a knowledge base $X$ with respect to another knowledge base $Y$, so that queries can be answered from theory prime implicates in polynomial time. 

Most of the work in knowledge
compilation have been restricted to propositional logic and first order logic in
spite of an increasing intrest in modal logic. Due to lack of expressive power in propositional logic and the undecidability of first order logic, modal logic is required as a knowledge representation language in many problems. Modal logic gives a trade-off between expressivity and complexity as they are more expressive than propositional logic and computationally better behaved than first order logic. An
algorithm to compute the set of prime implicates of modal logic $\mathcal{K}$ and $\mathcal{K}_{n}$ have been proposed in \cite{Bienvenu} and \cite{Bienvenu1} respectively. 

In \cite{Marquis}, the notion of prime implicates is generalized to theory prime implicates in propositional logic where the size of theory prime implicate compilation of a knowledge base is always exponentially smaller than the size of its prime implicate compilation. Moreover, query answering from theory prime implicate compilation can be performed in time polynomial in their size. In this paper we extend this concept from propositional to modal logic using the algorithm in \cite{Bienvenu}. So here we compute the theory prime implicate of a knowledge base $X$ with respect to another {\it restricted} knowledge base $\Box Y$, i.e, $\Theta(X,\Box Y)$ where $Y$ is a propositional knowledge base such that $X\models Y$. It can be noted that if $Y=\emptyset$ then $\Theta(X,\Box Y)$ becomes $\Pi(X)$.

The paper is organized as follows. In section 2 We give basic results in modal logic. In Section 3 we propose basic definitions of prime implicates, theory prime implicates and we describe the properties of theory prime implicates, the algorithm for computing theory prime implicates and query answering in modal logic. Section 4 concludes the paper.   

\section{Preliminaries}
Let us now discuss the basics of modal logic $\mathcal{K}$ from \cite{Blackburn1, Blackburn2}. The alphabet of modal formulas is $Var\cup\{\neg,\vee, \Diamond,(,)\}$. $Var$ is a countable set of variables denoted by $p,q,r,\ldots$. The connectives $\neg$ and $\vee$ are negation and disjunction. $\Diamond$ is the modal operator `possible'. The modal formulas $MF$ are defined inductively as follows. Variables are modal formulas. If $A$ and $B$ are modal formulas then $\neg A, A\vee B, \Diamond A$ are modal formulas. For the sake of convenience, we introduce the connectives $A\rightarrow B\equiv\neg A\vee B, A\leftrightarrow B\equiv(A\rightarrow B)\wedge(B\rightarrow A)$. The `necessary' operator $\Box$ is defined as $\Box A\equiv\neg\Diamond\neg A$. We avoid using parentheses whenever possible.

\begin{definition}\label{model}
The semantics of modal logic $\mathcal{K}$ is defined using Kripke models \cite{Kripke}. A Kripke model $M$ is a triple $\langle W, R, v\rangle$ where $W$ is a nonempty set (of worlds), $R$ is a binary relation on $W$ called the accessibility relation , so if $(w, w^{'})\in R$ then we say $w^{'}$ is accessible from $w$, and $v: W\rightarrow 2^{Var}$ is a valuation function, which assigns to each world $w\in W$ a subset $v(w)$ of $Var$ such that $p$ is true at a world $w$ iff $p\in v(w)$.
\end{definition}

\begin{definition}\label{model_truth}
Given any Kripke model $M=\langle W, R, v\rangle$, a world $w\in W$, and a formula $\phi\in MF$, the truth of $\phi$ at $w$ of $M$ denoted by $M,w\models\phi$, is defined inductively as follows:
\begin{itemize}
 \item $M,w\models p$ where $p\in Var$ iff  $p\in v(w)$,
 \item $M,w\models\neg\phi$ iff $M,w\not\models\phi$, 
 \item $M,w\models\phi\wedge\psi$ iff $M,w\models\phi$ and $M,w\models\psi$,
 \item $M,w\models\phi\vee\psi$ iff $M,w\models\phi$ or $M,w\models\psi$,
 \item $M,w\models\Box\phi$ iff for all $w^{'}\in W$ with $wRw^{'}$ we have $M,w^{'}\models\phi$,
 \item $M,w\models\Diamond\phi$ iff for some $w^{'}\in W$ with $wRw^{'}$ we have $M,w^{'}\models\phi$.
  \end{itemize}
\end{definition}

We say that a formula $\phi$ is satisfiable if there exists a model $M$ and and a world $w$ such that $M,w\models\phi$ and say $\phi$ is valid denoted by $\models\phi$ if $M,w\models\phi$ for all $M$ and $w$. A formula $\phi$ is unsatisfiable written as $\phi\models \bot$ if there exists no $M$ and $w$ for which $M,w\models\phi$. A formula $\psi$ is a logical consequence of a formula $\phi$ written as $\phi\models\psi$ if $M,w\models\phi$ implies $M,w\models\psi$ for every model $M$ and world $w\in W$.

There are two types of logical consequences in modal logic which are:

\begin{enumerate}
\item a formula $\psi$ is a global consequence of $\phi$ if whenever $M,w\models\phi$ for every world $w$ of a model $M$, then $M,w\models\psi$ for every world $w$ of $M$.
\item a formula $\psi$ is a local consequence of $\phi$ if $M,w\models\phi$ implies $M,w\models\psi$ for every model $M$ and world $w$.
\end{enumerate}

Eventhough both consequences exist, in this paper we will only study global consequences and whenever $\phi\models\psi$ we mean $\psi$ is a global consequence of $\phi$. 

Two formulas $\phi$ and $\psi$ are equivalent written as $\phi\equiv\psi$ or $\models\phi\leftrightarrow\psi$ if both $\phi\models\psi$ and $\psi\models\phi$. A formula $\phi$ is said to be logically stronger than $\psi$ or $\psi$ is said to be weaker than $\phi$ if $\phi\models\psi$ and $\psi\not\models\phi$. We can always strengthen a premise and weaken a consequence as $\phi\models\psi$ implies $\phi\wedge\chi\models\psi$ and $\phi\models\psi\vee\chi$ for some formula $\chi$.

It can be noted that in Definition \ref{model} and \ref{model_truth}, if we take $R$ to be a reflexive relation then system $\mathcal{K}$ becomes system $\mathcal{T}$. There are some results in this paper which holds in system $\mathcal{T}$ only. As any theorem of $\mathcal{K}$ is a theorem in $\mathcal{T}$ so every result holding in $\mathcal{K}$ also holds in $\mathcal{T}$.

The definitions of literals, clauses, terms and formulas in modal logic $\mathcal{T}$ known as definition $D4$ in \cite{Bienvenu} are given below.

\begin{definition}\label{LCTF}
The literals L, clauses C, terms T, and formulas F are defined as follows:

\hskip 2cm $L::=a\mid\neg a\mid\Box F\mid\Diamond F$

\hskip 2cm $C::=L\mid C\vee C$

\hskip 2cm $T::=L\mid T\wedge T$

\hskip 2cm $F::=a\mid\neg a\mid F\wedge F\mid F\vee F\mid \Box F\mid\Diamond F$

\end{definition}

A formula is said to be in conjunctive normal form (CNF) if it is a conjunction of clauses and it is in disjunctive normal form (DNF) if it is a disjunction of terms. The transformation of a formula to CNF or DNF is exponential in both time and space. The number of clauses in a CNF formula $\phi$ is denoted as $\mbox{nb\_cl}(\phi)$.
 
We now present some basic properties of logical consequences and equivalences in $\mathcal{K}$ which will be used in the proofs of some theorems in our paper.

\begin{lemma}\label{DiamondBox_equiv}
Let $\phi$ and $\psi$ be modal formulas. Then the following three statements are equivalent. 
\begin{enumerate}
 \item[(i)] $\phi\equiv\psi$ 
 \item[(ii)] $\Diamond\phi\equiv\Diamond\psi$ 
\item[(iii)] $\Box\phi\equiv\Box\psi$
\end{enumerate}
\end{lemma}

\begin{proof}
$(i)\Rightarrow(ii)$: Let~ $\phi\equiv\psi$. Then there is some $\mathcal{M}$ and $w$ such that $\mathcal{M},w\models\phi$ iff $\mathcal{M},w\models\psi$. Construct a new model $\mathcal{M}^{'}$ which contains the model  $\mathcal{M}$, state $w$ and an arrow from $w^{'}$ to $w$, so $\mathcal{M}^{'},w^{'}\models\Diamond\phi$ iff $\mathcal{M}^{'},w^{'}\models\Diamond\psi$. This implies $\Diamond\phi\equiv\Diamond\psi$.

$(ii)\Rightarrow(i)$: Let~ $\Diamond\phi\equiv\Diamond\psi$. Let there be a model $\mathcal{M}$ and a state $w$ such that $\mathcal{M},w\models\Diamond\phi$ iff $\mathcal{M},w\models\Diamond\psi$. Then there exists a state $w^{'}$ such that $Rww^{'}$ and $\mathcal{M},w^{'}\models\phi$ iff $\mathcal{M},w^{'}\models\psi$. This implies  $\phi\equiv\psi$. 

$(i)\Rightarrow(iii)$: Let~ $\phi\equiv\psi$. Then there is some $\mathcal{M}$ and $w$ such that $\mathcal{M},w\models\phi$ iff $\mathcal{M},w\models\psi$. Construct a new model $\mathcal{M}^{'}$ which contains the model  $\mathcal{M}, w$ and a relation $Rw^{'}w$ for each $w$ then so $\mathcal{M}^{'},w^{'}\models\Box\phi$ iff $\mathcal{M}^{'},w^{'}\models\Box\psi$. This implies $\Box\phi\equiv\Box\psi$.

$(iii)\Rightarrow(i)$: Let~ $\Box\phi\equiv\Box\psi$. Let there be a model $\mathcal{M}$ and a state $w$ such that $\mathcal{M},w\models\Box\phi$ iff $\mathcal{M},w\models\Box\psi$. Then for all state $w^{'}$ such that $Rww^{'}$ and $\mathcal{M},w^{'}\models\phi$ iff $\mathcal{M},w^{'}\models\psi$. This implies  $\phi\equiv\psi$. \hfill{$\Box$}  
\end{proof}

We now extend the definition of $\models$ with respect to a formula $Y$ written as $\models_{Y}$. 

\begin{definition}
Let $X_1, X_2$ be modal formulas and $Y$ be any propositional formula. We define $\models_{Y}$ over $MF\times MF$ (as the extension of $\models$) by $X_1\models_{Y}X_2$ iff $X_1\cup Y\models X_2$. When $X_1\models_{Y}X_2$ holds then we say that $X_2$ is a $Y$-logical consequence of $X_1$. We define the equivalence relation $\equiv_{Y}$ over $MF$ by $X_1\equiv_{Y}X_2$ iff $X_1\models_{Y}X_2$ and $X_2\models_{Y}X_1$. When $X_1\equiv_{Y}X_2$ holds we say $X_1$ and $X_2$ are $Y$-equivalent. 
\end{definition}

 We now present the following lemmas which will be used in the proofs of Theorem \ref{weaken_the_consequence} and Theorem \ref{tpi_size_less_than_pi} later.

\begin{lemma}\label{symmetry_propositional}
Let $\psi$ and $\chi$ be modal formulas, and $Y$ be any propositional formula. Then the following three statements are equivalent with respect to $\models_{Y}$.
\begin{enumerate}
\item[(i)] $\psi\models_{Y} \chi$
\item[(ii)] $\models_{Y}\neg\psi\vee\chi$
\item[(iii)] $\psi\wedge\neg\chi\models_{Y}\bot$
\end{enumerate}
\end{lemma}
\begin{proof}
 $(i)\Rightarrow (ii)$:  Let $\mathcal{M}=\langle W,R,v\rangle$ be a model and $w$ be a state in $W$. Let $\psi\models_{Y}\chi$. Then $\psi\wedge Y\models\chi$. So if $\mathcal{M}, w\models\psi$ and $\mathcal{M},w\models Y$ then $\mathcal{M},w\models\chi$ for all $\mathcal{M}$ and $w$. As $p\rightarrow q\equiv\neg p\vee q$, so $\mathcal{M}, w\not\models\psi$ or $\mathcal{M},w\not\models Y$ or $\mathcal{M},w\models\chi$ for all $\mathcal{M}$ and $w$. So $\mathcal{M},w\not\models Y$ or  $(\mathcal{M}, w\not\models\psi$ or $\mathcal{M},w\models\chi)$ for all $\mathcal{M}$ and $w$. This implies, if $\mathcal{M},w\models Y$ then  $(\mathcal{M}, w\not\models\psi$ or $\mathcal{M},w\models\chi)$ for all $\mathcal{M}$ and $w$. If $\mathcal{M},w\models Y$ then $(\mathcal{M}, w\models\neg\psi$ or $\mathcal{M},w\models\chi)$ for all $\mathcal{M}$ and $w$. If $\mathcal{M},w\models Y$ then $\mathcal{M},w\models\neg\psi\vee\chi$ for all $\mathcal{M}$ and $w$. So $Y\models \neg\psi\vee\chi$. Hence, $\models_{Y}\neg\psi\vee\chi$.

$(ii)\Rightarrow (iii)$: If $\mathcal{M},w\models Y$ then $\mathcal{M},w\models\neg\psi\vee\chi$ for all $\mathcal{M}$ and $w$. So, if $\mathcal{M},w\models Y$ then $\mathcal{M},w\not\models\neg(\neg\psi\vee\chi)$ for all $\mathcal{M}$ and $w$. If $\mathcal{M},w\models Y$ then $\mathcal{M},w\not\models(\psi\wedge\neg\chi)$ for all $\mathcal{M}$ and $w$. $\mathcal{M},w\not\models Y$ or $\mathcal{M},w\not\models(\psi\wedge\neg\chi)$ for all $\mathcal{M}$ and $w$. Hence, $\neg(\mathcal{M},w\models Y \mbox{~and~} \mathcal{M},w\models\psi\wedge\neg\chi)$ holds for all $\mathcal{M}$ and $w$. This implies, $\neg(\mathcal{M},w\models Y\wedge\psi\wedge\neg\chi)$ holds for all $\mathcal{M}$ and $w$. $\mathcal{M},w\not\models Y\wedge\psi\wedge\neg\chi$ for all $\mathcal{M}$ and $w$. So, $Y\wedge\psi\wedge\neg\chi\models\bot$. This implies $\psi\wedge\neg\chi\models_{Y}\bot$.

$(iii)\Rightarrow (i)$: Let $\psi\wedge\neg\chi\models_{Y}\bot$. This implies $Y\wedge\psi\wedge\neg\chi\models\bot$. So, $Y\wedge\psi\models\chi$. Hence, $\psi\models_{Y}\chi$. \hfill{$\Box$}
\end{proof}

\begin{lemma}\label{symmetry_entailment} 
Let $\psi$ and $\chi$ be modal formulas, and $Y$ be any propositional formula. Then the following three statements are equivalent in modal system $\mathcal{K}$.
\begin{enumerate}
\item[(i)] $\psi\models_{Y}\chi$.
\item[(ii)] $\Diamond\psi\models_{\Box Y}\Diamond\chi$.
\item[(iii)] $\Box\psi\models_{\Box Y}\Box\chi$.
\end{enumerate}
\end{lemma}
\begin{proof} Let $\mathcal{M}=\langle W,R,v\rangle$ be a model and $w$ be a state in $W$.

$(i)\Rightarrow (ii)$: Let $\Diamond\psi\not\models_{\Box Y}\Diamond\chi$. So $\Diamond\psi\wedge\Box Y\not\models\Diamond\chi$. $\Diamond\psi\wedge\Box Y\wedge\neg\Diamond\chi\not\models\bot$. Then there exists $\mathcal{M}$ and $w$ such that $\mathcal{M},w\models \Diamond\psi\wedge\Box Y\wedge\neg\Diamond\chi$. So, $\mathcal{M},w\models \Diamond\psi\wedge\Box Y\wedge\Box\neg\chi$. $\mathcal{M},w\models \Diamond\psi$ and $\mathcal{M},w\models\Box Y$ and $\mathcal{M},w\models\Box\neg\chi$. Then for all  state $w^{'}$ such that $Rww^{'}$ and $\mathcal{M}, w^{'}\models\psi$, $\mathcal{M}, w^{'}\models Y$, and $\mathcal{M}, w^{'}\models\neg\chi$. So $\psi\wedge Y\wedge\neg\chi\not\models\bot$. $\psi\wedge Y\not\models\chi$. Hence, $\psi\not\models_{Y}\chi$.  

$(ii)\Rightarrow(i)$: Let $\psi\not\models_{Y}\chi$. $\psi\wedge Y\not\models\chi$. Then there exists a model $\mathcal{M}$ and a state $w$ such that $\mathcal{M},w\models\psi\wedge Y$ and $\mathcal{M},w\not\models\chi$. So, $\mathcal{M},w\models\psi\wedge Y$ and $\mathcal{M},w\models\neg\chi$. $\mathcal{M},w\models\psi$, $\mathcal{M},w\models Y$ and $\mathcal{M},w\models\neg\chi$. Let us create a new model $\mathcal{M}^{'}$ by adding a new world $w^{'}$ and an arrow from $w^{'}$ to each $w$. Then, $\mathcal{M}^{'},w^{'}\models\Diamond\psi$, $\mathcal{M}^{'},w^{'}\models \Box Y$ and $\mathcal{M}^{'},w^{'}\models\Box\neg\chi$. So, $\mathcal{M}^{'},w^{'}\models\Diamond\psi\wedge\Box Y\wedge\Box\neg\chi$. $\Diamond\psi\wedge\Box Y\wedge\Box\neg\chi\not\models\bot$.  $\Diamond\psi\wedge\Box Y\not\models\neg\Box\neg\chi$. $\Diamond\psi\wedge\Box Y\not\models\Diamond\chi$. So $\Diamond\psi\not\models_{\Box Y}\Diamond\chi$.

$(i)\Rightarrow(iii)$: It is similar to the proof of $(i)\Rightarrow(ii)$.

$(iii)\Rightarrow(i)$: Let $\Box\psi\models_{\Box Y}\Box\chi$. So, $\Box\psi\wedge\Box Y\models\Box\chi$. This implies $\neg\Box\chi\models\neg(\Box\psi\wedge\Box Y)$. So, $\neg\Box\chi\models\neg\Box\psi\vee\neg\Box Y$. This implies, $\Diamond\neg\chi\models\Diamond\neg\psi\vee\Diamond\neg Y$. By property of $\mathcal{K}$ , we have $\Diamond\neg\chi\models\Diamond(\neg\psi\vee\neg Y)$. Again by Lemma \ref{DiamondBox_equiv}, we get $\neg\chi\models\neg\psi\vee\neg Y$. So $\neg\chi\models\neg(\psi\wedge Y)$. This implies $\psi\wedge Y\models\chi$. So, $\psi\models_{Y}\chi$. \hfill{$\Box$}
\end{proof}

The following lemma which holds in $\mathcal{K}$ is used in the proof of Theorem \ref{weaken_the_consequence} and \ref{tpi_size_less_than_pi}.

\begin{theorem}\label{several_possibilities_wrt_boxY}
Let $\beta_1,\beta_2,\ldots,\beta_m,\gamma_1,\gamma_2,\ldots,\gamma_n, \phi_1,\phi_2,\ldots,\phi_q,\xi_1,\xi_2,\ldots,\xi_r$ be modal formulas and $\alpha_1,\alpha_2,\ldots,\alpha_l, \psi_1,\psi_2,\ldots,\psi_p, Y$ be propositional formulas. Then
$(\vee_{i=1}^{l}\alpha_i)\wedge(\vee_{j=1}^{m}\Diamond\beta_j)\wedge(\vee_{k=1}^{n}\Box\gamma_{k})\wedge((\vee_{i=1}^{l}\alpha_i)\vee(\vee_{j=1}^{m}\Diamond\beta_j))\wedge((\vee_{i=1}^{l}\alpha_i)\vee(\vee_{k=1}^{n}\Box\gamma_{k}))\wedge((\vee_{j=1}^{m}\Diamond\beta_j)\vee(\vee_{k=1}^{n}\Box\gamma_{k}))\wedge((\vee_{i=1}^{l}\alpha_i)\vee(\vee_{j=1}^{m}\Diamond\beta_j)\vee(\vee_{k=1}^{n}\Box\gamma_{k}))\wedge\psi_1\wedge\ldots\wedge\psi_p\wedge\Box\phi_1\wedge\ldots\wedge\Box\phi_q\wedge\Diamond\xi_1\wedge\ldots\wedge\Diamond\xi_r\models_{\Box Y}\bot$ if and only if
\begin{enumerate}
\item $(\vee_{i=1}^{l}\alpha_i)\wedge\psi_1\wedge\ldots\wedge\psi_p\models_{Y}\bot$ or
\item $(\vee_{j=1}^{m}\beta_j)\wedge\phi_1\wedge\ldots\wedge\phi_q\models_{Y}\bot$ or
\item $(\vee_{k=1}^{n}\gamma_{k})\wedge\xi_u\wedge\phi_1\wedge\ldots\wedge\phi_q\models_{Y}\bot$ for $1\leq u\leq r$ or
\item $((\vee_{i=1}^{l}\alpha_i)\vee(\vee_{j=1}^{m}\beta_j))\wedge\phi_1\wedge\ldots\wedge\phi_q\models_{Y}\bot$ or
\item $((\vee_{i=1}^{l}\alpha_i)\vee(\vee_{k=1}^{n}\gamma_{k}))\wedge\xi_u\wedge\phi_1\wedge\ldots\wedge\phi_q\models_{Y}\bot$ for $1\leq u\leq r$ or
\item $((\vee_{j=1}^{m}\beta_j)\vee(\vee_{k=1}^{n}\gamma_{k}))\wedge\phi_1\wedge\ldots\wedge\phi_q\models_{Y}\bot$ or
\item $((\vee_{i=1}^{l}\alpha_i)\vee(\vee_{j=1}^{m}\beta_j)\vee(\vee_{k=1}^{n}\gamma_{k}))\wedge\phi_1\wedge\ldots\wedge\phi_q\models_{Y}\bot$.
\end{enumerate}
\end{theorem}
\begin{proof}
Suppose 
\begin{enumerate}
\item $(\vee_{i=1}^{l}\alpha_i)\wedge\psi_1\wedge\ldots\wedge\psi_p\not\models_{Y}\bot$,  
 \item $(\vee_{j=1}^{m}\beta_j)\wedge\phi_1\wedge\ldots\wedge\phi_q\not\models_{Y}\bot$, 
 \item $(\vee_{k=1}^{n}\gamma_{k})\wedge\xi_u\wedge\phi_1\wedge\ldots\wedge\phi_q\not\models_{Y}\bot$ for $1\leq u\leq r$,
\item $((\vee_{i=1}^{l}\alpha_i)\vee(\vee_{j=1}^{m}\beta_j))\wedge\phi_1\wedge\ldots\wedge\phi_q\not\models_{Y}\bot$,
 \item $((\vee_{i=1}^{l}\alpha_i)\vee(\vee_{k=1}^{n}\gamma_{k}))\wedge\xi_u\wedge\phi_1\wedge\ldots\wedge\phi_q\not\models_{Y}\bot$ for $1\leq u\leq r$,
 \item $((\vee_{j=1}^{m}\beta_j)\vee(\vee_{k=1}^{n}\gamma_{k}))\wedge\phi_1\wedge\ldots\wedge\phi_q\not\models_{Y}\bot$, and
\item $((\vee_{i=1}^{l}\alpha_i)\vee(\vee_{j=1}^{m}\beta_j)\vee(\vee_{k=1}^{n}\gamma_{k}))\wedge\phi_1\wedge\ldots\wedge\phi_q\not\models_{Y}\bot$. 
\end{enumerate}

Then there exists a model $\mathcal{M}^{'}$ and a world $w^{'}$ such that
\begin{enumerate}
\item $\mathcal{M}^{'}, w^{'}\models(\vee_{i=1}^{l}\alpha_i)\wedge\psi_1\wedge\ldots\wedge\psi_p\wedge Y$
\item $\mathcal{M}^{'},w^{'}\models(\vee_{j=1}^{m}\beta_j)\wedge\phi_1\wedge\ldots\wedge\phi_q\wedge Y$, 
 \item $\mathcal{M}^{'},w^{'}\models(\vee_{k=1}^{n}\gamma_{k})\wedge\xi_u\wedge\phi_1\wedge\ldots\wedge\phi_q\wedge Y$ for $1\leq u\leq r$,
 \item $\mathcal{M}^{'},w^{'}\models((\vee_{i=1}^{l}\alpha_i)\vee(\vee_{j=1}^{m}\beta_j))\wedge\phi_1\wedge\ldots\wedge\phi_q\wedge Y$,
 \item $\mathcal{M}^{'},w^{'}\models((\vee_{i=1}^{l}\alpha_i)\vee(\vee_{k=1}^{n}\gamma_{k}))\wedge\xi_u\wedge\phi_1\wedge\ldots\wedge\phi_q\wedge Y$ for $1\leq u\leq r$,
 \item $\mathcal{M}^{'},w^{'}\models((\vee_{j=1}^{m}\beta_j)\vee(\vee_{k=1}^{n}\gamma_{k}))\wedge\phi_1\wedge\ldots\wedge\phi_q\wedge Y$, and
\item $\mathcal{M}^{'},w^{'}\models((\vee_{i=1}^{l}\alpha_i)\vee(\vee_{j=1}^{m}\beta_j)\vee(\vee_{k=1}^{n}\gamma_{k}))\wedge\phi_1\wedge\ldots\wedge\phi_q\wedge Y$.
\end{enumerate}

Hence, we obtain the followings from the above statements: 
\begin{enumerate}
\item As $\mathcal{M}^{'}, w^{'}\models(\vee_{i=1}^{l}\alpha_i)\wedge\psi_1\wedge\ldots\wedge\psi_p$ and $\mathcal{M}^{'}, w^{'}\models Y$, let there be a propositional model w of $(\vee_{i=1}^{l}\alpha_i)\wedge\psi_1\wedge\ldots\wedge\psi_p$ such that $\mathcal{M}^{'}, w\models(\vee_{i=1}^{l}\alpha_i)\wedge\psi_1\wedge\ldots\wedge\psi_p$ and $\mathcal{M}^{'}, w^{'}\models Y$. Construct a new model $\mathcal{M}$ which contains the model $\mathcal{M}^{'}$, $w$ and a relation $Rww^{'}$ for each $w^{'}$ then $\mathcal{M}, w\models(\vee_{i=1}^{l}\alpha_i)\wedge\psi_1\wedge\ldots\wedge\psi_p$ and $\mathcal{M}, w\models \Box Y$. So, $\mathcal{M}, w\models(\vee_{i=1}^{l}\alpha_i)\wedge\psi_1\wedge\ldots\wedge\psi_p\wedge\Box Y$.
\item Construct a new model $\mathcal{M}$ which contains the model $\mathcal{M}^{'}$, state $w^{'}$ and add a new world $w$ and a relation $Rww^{'}$ for each $w^{'}$ then $\mathcal{M},w\models(\vee_{j=1}^{m}\Diamond\beta_j)\wedge\Box\phi_1\wedge\ldots\wedge\Box\phi_q\wedge \Box Y$.
\item Similarly like (2) we have $\mathcal{M},w\models(\vee_{k=1}^{n}\Box\gamma_{k})\wedge\Diamond\xi_u\wedge\Box\phi_1\wedge\ldots\wedge\Box\phi_q\wedge \Box Y$ for $1\leq u\leq r$.
\item  As $\mathcal{M}^{'},w^{'}\models(\vee_{i=1}^{l}\alpha_i)\wedge\phi_1\wedge\ldots\wedge\phi_q\wedge Y$ or $\mathcal{M}^{'},w^{'}\models (\vee_{j=1}^{m}\beta_j)\wedge\phi_1\wedge\ldots\wedge\phi_q\wedge Y$. Let there be a propositional model $w$ of $(\vee_{i=1}^{l}\alpha_i)$ so $(\mathcal{M}^{'},w\models(\vee_{i=1}^{l}\alpha_i)\wedge\mathcal{M}^{'},w^{'}\models\phi_1\wedge\ldots\wedge\phi_q\wedge Y)\vee\mathcal{M}^{'},w^{'}\models (\vee_{j=1}^{m}\beta_j)\wedge\phi_1\wedge\ldots\wedge\phi_q\wedge Y$. Construct a new model $\mathcal{M}$ which contains the model $\mathcal{M}^{'}$, $w$ and a relation $Rww^{'}$ for each $w^{'}$ then  $(\mathcal{M},w\models(\vee_{i=1}^{l}\alpha_i)\wedge\mathcal{M},w\models\Box\phi_1\wedge\ldots\wedge\Box\phi_q\wedge \Box Y)\vee\mathcal{M},w\models (\vee_{j=1}^{m}\Diamond\beta_j)\wedge\Box\phi_1\wedge\ldots\wedge\Box\phi_q\wedge \Box Y$. So $(\mathcal{M},w\models(\vee_{i=1}^{l}\alpha_i)\wedge \Box\phi_1\wedge\ldots\wedge\Box\phi_q\wedge \Box Y)$ or  $(\mathcal{M},w\models (\vee_{j=1}^{m}\Diamond\beta_j)\wedge\Box\phi_1\wedge\ldots\wedge\Box\phi_q\wedge \Box Y)$. By distributivity, $(\mathcal{M},w\models((\vee_{i=1}^{l}\alpha_i)\vee(\vee_{j=1}^{m}\Diamond\beta_j))\wedge \Box\phi_1\wedge\ldots\wedge\Box\phi_q\wedge \Box Y$.
\item Similarly like (4) we have $\mathcal{M},w\models((\vee_{i=1}^{l}\alpha_i)\vee(\vee_{k=1}^{n}\Box\gamma_{k}))\wedge\Diamond\xi_u\wedge\Box\phi_1\wedge\ldots\wedge\Box\phi_q\wedge\Box Y$ for $1\leq u\leq r$,
\item Similarly like (2) or (3) we have $\mathcal{M},w\models((\vee_{j=1}^{m}\Diamond\beta_j)\vee(\vee_{k=1}^{n}\Box\gamma_{k}))\wedge\Box\phi_1\wedge\ldots\wedge\Box\phi_q\wedge\Box Y$.
\item Similarly like (4) we have $\mathcal{M},w\models((\vee_{i=1}^{l}\alpha_i)\vee(\vee_{j=1}^{m}\Diamond\beta_j)\vee(\vee_{k=1}^{n}\Box\gamma_{k}))\wedge\Box\phi_1\wedge\ldots\wedge\Box\phi_q\wedge\Box Y$.
\end{enumerate}

From (1) to (7) we get $\mathcal{M},w\models(\vee_{i=1}^{l}\alpha_i)\wedge(\vee_{j=1}^{m}\Diamond\beta_j)\wedge(\vee_{k=1}^{n}\Box\gamma_{k})\wedge((\vee_{i=1}^{l}\alpha_i)\vee(\vee_{j=1}^{m}\Diamond\beta_j))\wedge((\vee_{i=1}^{l}\alpha_i)\vee(\vee_{k=1}^{n}\Box\gamma_{k}))\wedge((\vee_{j=1}^{m}\Diamond\beta_j)\vee(\vee_{k=1}^{n}\Box\gamma_{k}))\wedge((\vee_{i=1}^{l}\alpha_i)\vee(\vee_{j=1}^{m}\Diamond\beta_j)\vee(\vee_{k=1}^{n}\Box\gamma_{k}))\wedge\psi_1\wedge\ldots\wedge\psi_p\wedge\Box\phi_1\wedge\ldots\wedge\Box\phi_q\wedge\Diamond\xi_1\wedge\ldots\wedge\Diamond\xi_r\wedge\Box Y$. This implies, 
$(\vee_{i=1}^{l}\alpha_i)\wedge(\vee_{j=1}^{m}\Diamond\beta_j)\wedge(\vee_{k=1}^{n}\Box\gamma_{k})\wedge((\vee_{i=1}^{l}\alpha_i)\vee(\vee_{j=1}^{m}\Diamond\beta_j))\wedge((\vee_{i=1}^{l}\alpha_i)\vee(\vee_{k=1}^{n}\Box\gamma_{k}))\wedge((\vee_{j=1}^{m}\Diamond\beta_j)\vee(\vee_{k=1}^{n}\Box\gamma_{k}))\wedge((\vee_{i=1}^{l}\alpha_i)\vee(\vee_{j=1}^{m}\Diamond\beta_j)\vee(\vee_{k=1}^{n}\Box\gamma_{k}))\wedge\psi_1\wedge\ldots\wedge\psi_p\wedge\Box\phi_1\wedge\ldots\wedge\Box\phi_q\wedge\Diamond\xi_1\wedge\ldots\wedge\Diamond\xi_r\not\models_{\Box Y}\bot$.

Conversely, Suppose $(\vee_{i=1}^{l}\alpha_i)\wedge(\vee_{j=1}^{m}\Diamond\beta_j)\wedge(\vee_{k=1}^{n}\Box\gamma_{k})\wedge((\vee_{i=1}^{l}\alpha_i)\vee(\vee_{j=1}^{m}\Diamond\beta_j))\wedge((\vee_{i=1}^{l}\alpha_i)\vee(\vee_{k=1}^{n}\Box\gamma_{k}))\wedge((\vee_{j=1}^{m}\Diamond\beta_j)\vee(\vee_{k=1}^{n}\Box\gamma_{k}))\wedge((\vee_{i=1}^{l}\alpha_i)\vee(\vee_{j=1}^{m}\Diamond\beta_j)\vee(\vee_{k=1}^{n}\Box\gamma_{k}))\wedge\psi_1\wedge\ldots\wedge\psi_p\wedge\Box\phi_1\wedge\ldots\wedge\Box\phi_q\wedge\Diamond\xi_1\wedge\ldots\wedge\Diamond\xi_r\not\models_{\Box Y}\bot$. Then there exists a model $\mathcal{M}$ and a state $w$ such that 
$\mathcal{M},w\models(\vee_{i=1}^{l}\alpha_i)\wedge(\vee_{j=1}^{m}\Diamond\beta_j)\wedge(\vee_{k=1}^{n}\Box\gamma_{k})\wedge((\vee_{i=1}^{l}\alpha_i)\vee(\vee_{j=1}^{m}\Diamond\beta_j))\wedge((\vee_{i=1}^{l}\alpha_i)\vee(\vee_{k=1}^{n}\Box\gamma_{k}))\wedge((\vee_{j=1}^{m}\Diamond\beta_j)\vee(\vee_{k=1}^{n}\Box\gamma_{k}))\wedge((\vee_{i=1}^{l}\alpha_i)\vee(\vee_{j=1}^{m}\Diamond\beta_j)\vee(\vee_{k=1}^{n}\Box\gamma_{k}))\wedge\psi_1\wedge\ldots\wedge\psi_p\wedge\Box\phi_1\wedge\ldots\wedge\Box\phi_q\wedge\Diamond\xi_1\wedge\ldots\wedge\Diamond\xi_r\wedge\Box Y$. Then the following must hold. 
\begin{enumerate}
\item $\mathcal{M}, w\models(\vee_{i=1}^{l}\alpha_i)\wedge\psi_1\wedge\ldots\wedge\psi_p\wedge \Box Y$
\item $\mathcal{M},w\models(\vee_{j=1}^{m}\Diamond\beta_j)\wedge\Box\phi_1\wedge\ldots\wedge\Box\phi_q\wedge \Box Y$, 
 \item $\mathcal{M},w\models(\vee_{k=1}^{n}\Box\gamma_{k})\wedge\Diamond\xi_u\wedge\Box\phi_1\wedge\ldots\wedge\Box\phi_q\wedge \Box Y$ for $1\leq u\leq r$,
 \item $\mathcal{M},w\models((\vee_{i=1}^{l}\alpha_i)\vee(\vee_{j=1}^{m}\Diamond\beta_j))\wedge\Box\phi_1\wedge\ldots\wedge\Box\phi_q\wedge \Box Y$,
 \item $\mathcal{M},w\models((\vee_{i=1}^{l}\alpha_i)\vee(\vee_{k=1}^{n}\Box\gamma_{k}))\wedge\Diamond\xi_u\wedge\Box\phi_1\wedge\ldots\wedge\Box\phi_q\wedge \Box Y$ for $1\leq u\leq r$,
 \item $\mathcal{M},w\models((\vee_{j=1}^{m}\Diamond\beta_j)\vee(\vee_{k=1}^{n}\Box\gamma_{k}))\wedge\Box\phi_1\wedge\ldots\wedge\Box\phi_q\wedge \Box Y$, and
\item $\mathcal{M},w\models((\vee_{i=1}^{l}\alpha_i)\vee(\vee_{j=1}^{m}\Diamond\beta_j)\vee(\vee_{k=1}^{n}\Box\gamma_{k}))\wedge\Box\phi_1\wedge\ldots\wedge\Box\phi_q\wedge \Box Y$.
\end{enumerate}
Hence, we get the following from above statements:
\begin{enumerate}
\item $\mathcal{M}, w\models(\vee_{i=1}^{l}\alpha_i)\wedge\psi_1\wedge\ldots\wedge\psi_p$ and $\mathcal{M}, w\models\Box Y$. Let there be a propositional model $w^{'}$ of $(\vee_{i=1}^{l}\alpha_i)\wedge\psi_1\wedge\ldots\wedge\psi_p$ so that $\mathcal{M}, w^{'}\models(\vee_{i=1}^{l}\alpha_i)\wedge\psi_1\wedge\ldots\wedge\psi_p$ and $\mathcal{M}, w\models\Box Y$. Construct a new Kripke model $\mathcal{M}^{'}$ which contains the model $\mathcal{M}$ and the world $w^{'}$ and $Rww^{'}$, so we get $\mathcal{M}^{'}, w^{'}\models(\vee_{i=1}^{l}\alpha_i)\wedge\psi_1\wedge\ldots\wedge\psi_p$ and $\mathcal{M}^{'}, w^{'}\models Y$. So $\mathcal{M}^{'}, w^{'}\models(\vee_{i=1}^{l}\alpha_i)\wedge\psi_1\wedge\ldots\wedge\psi_p\wedge Y$. $(\vee_{i=1}^{l}\alpha_i)\wedge\psi_1\wedge\ldots\wedge\psi_p\wedge Y\not\models\bot$. Hence, $(\vee_{i=1}^{l}\alpha_i)\wedge\psi_1\wedge\ldots\wedge\psi_p\not\models_{Y}\bot$.
\item $(\vee_{j=1}^{m}\beta_j)\wedge\phi_1\wedge\ldots\wedge\phi_q\wedge Y$ is satisfiable because there exists $w^{'}$ such that $Rww^{'}$ and $\mathcal{M}, w^{'}\models(\vee_{j=1}^{m}\beta_j)\wedge\phi_1\wedge\ldots\wedge\phi_q\wedge Y$. So, $(\vee_{j=1}^{m}\beta_j)\wedge\phi_1\wedge\ldots\wedge\phi_q\wedge Y\not\models\bot$. This implies $(\vee_{j=1}^{m}\beta_j)\wedge\phi_1\wedge\ldots\wedge\phi_q\not\models_{Y}\bot$.
\item Similarly, we can show $(\vee_{k=1}^{n}\gamma_{k})\wedge\xi_u\wedge\phi_1\wedge\ldots\wedge\phi_q \not\models_{Y}\bot$ for $1\leq u\leq r$ like (2).
\item As $\mathcal{M},w\models(\vee_{i=1}^{l}\alpha_i)\wedge\Box\phi_1\wedge\ldots\wedge\Box\phi_q\wedge \Box Y$ and $\mathcal{M},w\models(\vee_{j=1}^{m}\Diamond\beta_j)\wedge\Box\phi_1\wedge\ldots\wedge\Box\phi_q\wedge \Box Y$, so $(\mathcal{M},w\models(\vee_{i=1}^{l}\alpha_i)\wedge\mathcal{M},w\models\Box\phi_1\wedge\ldots\wedge\Box\phi_q\wedge \Box Y)\vee\mathcal{M},w\models(\vee_{j=1}^{m}\Diamond\beta_j)\wedge\Box\phi_1\wedge\ldots\wedge\Box\phi_q\wedge \Box Y$. Let there be a propositional model $w^{'}$ of $(\vee_{i=1}^{l}\alpha_i)$ such that $(\mathcal{M},w^{'}\models(\vee_{i=1}^{l}\alpha_i)\wedge\mathcal{M},w\models\Box\phi_1\wedge\ldots\wedge\Box\phi_q\wedge \Box Y)\vee\mathcal{M},w\models(\vee_{j=1}^{m}\Diamond\beta_j)\wedge\Box\phi_1\wedge\ldots\wedge\Box\phi_q\wedge \Box Y$. Construct a new Kripke model $\mathcal{M}^{'}$ which contains the model $\mathcal{M}$ and the world $w^{'}$ and $Rww^{'}$, so we get $(\mathcal{M}^{'},w^{'}\models(\vee_{i=1}^{l}\alpha_i)\wedge\mathcal{M}^{'},w^{'}\models\phi_1\wedge\ldots\wedge\phi_q\wedge Y)\vee\mathcal{M}^{'},w^{'}\models(\vee_{j=1}^{m}\beta_j)\wedge\phi_1\wedge\ldots\wedge\phi_q\wedge Y$. So, $\mathcal{M}^{'},w^{'}\models(\vee_{i=1}^{l}\alpha_i)\wedge\phi_1\wedge\ldots\wedge\phi_q\wedge Y$ or $\mathcal{M}^{'},w^{'}\models(\vee_{j=1}^{m}\beta_j)\wedge\phi_1\wedge\ldots\wedge\phi_q\wedge Y$. By distributivity, $\mathcal{M}^{'},w^{'}\models((\vee_{i=1}^{l}\alpha_i)\vee(\vee_{j=1}^{m}\beta_j))\wedge\phi_1\wedge\ldots\wedge\phi_q\wedge Y$. $((\vee_{i=1}^{l}\alpha_i)\vee(\vee_{j=1}^{m}\beta_j))\wedge\phi_1\wedge\ldots\wedge\phi_q\wedge Y\not\models\bot$. This implies, $((\vee_{i=1}^{l}\alpha_i)\vee(\vee_{j=1}^{m}\beta_j))\wedge\phi_1\wedge\ldots\wedge\phi_q\not\models_{Y}\bot$. 
\item Similarly we can show $((\vee_{i=1}^{l}\alpha_i)\vee(\vee_{k=1}^{n}\gamma_{k}))\wedge\xi_u\wedge\phi_1\wedge\ldots\wedge\phi_q\not\models_{Y}\bot$ for $1\leq u\leq r$ like (4). 

\item Similarly we can show $((\vee_{j=1}^{m}\beta_j)\vee(\vee_{k=1}^{n}\gamma_{k}))\wedge\phi_1\wedge\ldots\wedge\phi_q\not\models_{Y}\bot$ like (2) or (3). 
\item Similarly we can show $((\vee_{i=1}^{l}\alpha_i)\vee(\vee_{j=1}^{m}\beta_j)\vee(\vee_{k=1}^{n}\gamma_{k}))\wedge\phi_1\wedge\ldots\wedge\phi_q\not\models_{Y}\bot$ like (4). \hfill{$\Box$}
\end{enumerate}
\end{proof}

\section{Theory Prime Implicates}

 Now we give the definitions of prime implicates and prime implicants of a knowledge base $X$ with respect to $\models$ in modal logic.

\begin{definition}
A clause $C$ is said to be an implicate of a formula $X$ if $X\models C$. A clause $C$ is a prime implicate of $X$ if $C$ is an implicate of $X$ and there is no other implicate $C^{'}$ of $X$ such that $C^{'}\models C$. The set of prime implicates of $X$ is denoted by  $\Pi(X)$. 
\end{definition}

\begin{definition}
A term $C$ is said to be an implicant of a formula $X$ if $C\models X$. A term $C$ is said to be a prime implicant of $X$ if $C$ is an implicant of $X$ and and there is no other implicant $C^{'}$ of $X$ such that $C\models C^{'}$.
\end{definition}

\begin{definition}
A clause $C^{'}\in X$ is a minimal element of $X$ if for all $C\in X$, $C\models C^{'}$ implies $C\equiv C^{'}$. Similarly, a clause $C^{'}\in X$ is a minimal element of $X$ with respect to a propositional formula $Y$ if for all $C\in X$, $C\models_{Y} C^{'}$ implies $C\equiv_{Y} C^{'}$.
\end{definition}

So we note that prime implicates (or prime implicants) of a knowledge base $X$ are minimal elements with respect to $\models$ among the implicates (or implicants) of $X$ respectively.

We now extend the definition of prime implicate to theory prime implicate with respect to $\models_{Y}$ as follows.

\begin{definition}
Let $X$ and $Y$ be any modal formulas. A clause $C$ is a theory implicate of $X$ with respect to $Y$ iff $X\models_{Y}C$. A clause $C$ is a theory prime implicate of $X$ with respect to $Y$ iff $C$ is a theory implicate of $X$ with respect to $Y$ and there is no theory implicate $C^{'}$ of $X$ with respect to $Y$ such that $C^{'}\models _{Y}C$. We denote $\Theta(X,Y)$ as the set of theory prime implicates of $X$ with respect to $Y$.
\end{definition}

We note that the set of theory prime implicates of $X$ with respect to $Y$,i.e,  $\Theta(X,Y)$, is the minimal elements with respect to $\models_{Y}$ among the set of theory implicates of $X$ with respect to $Y$.

In the rest of the paper we compute the theory prime implicates of $X$ with respect to $\Box Y$ where $Y$ is a propositional formula using the above definitions and results. 
We have been able to compute theory prime implicates of $X$ with respect to a restricted modal knowledge base $\Box Y$, instead of an arbitrary modal formula $Z$. 

\subsection{Properties of Theory Prime Implicates}

Below we list some of the properties of theory prime implicates. 

\begin{lemma}\label{tpi_pi_subset}
Let $X$ be a modal formula and $Y$ be any propositional formula. Then $\Theta(X,\Box Y)\subseteq\Pi(X\cup\Box Y)$.
\end{lemma}
\begin{proof}
Let $C\in\Theta(X,\Box Y)$. So $X\models_{\Box Y} C$ and there is no theory implicate $C^{'}$ of $X$ with respect to $\Box Y$ such that $C^{'}\models_{\Box Y} C$. This implies $X\cup \Box Y\models C$. So $C$ is an implicate of $X\cup\Box Y$. If $C$ is not a prime implicate of $X\cup\Box Y$, i.e, $C\not\in\Pi(X\cup\Box Y)$ then there is an implicate $C^{'}$ of $X\cup\Box Y$ such that $C^{'}\models C$, which implies $C^{'}\models_{\Box Y}C$. As $C^{'}$ is an implicate of $X\cup\Box Y$ so $X\cup\Box Y\models C^{'}$, i.e, $X\models_{\Box Y} C^{'}$. So $C^{'}$ is a theory implicate of $X$ with respect to $\Box Y$. So we are getting a theory implicate $C^{'}$ of $X$ with respect to $\Box Y$ such that $C^{'}\models_{\Box Y}C$. This implies $C$ cannot be a theory prime implicate of $X$ with respect to $\Box Y$. Hence a contradiction. So $C\in\Pi(X\cup\Box Y)$. \hfill{$\Box$} 
\end{proof}
\begin{lemma}\label{tpi_pi_notbelongto}
If $C_1, C_2\in\Pi(X\cup\Box Y)$ and $C_1\models_{\Box Y}C_2$ then $C_2\not\in\Theta(X,\Box Y)$. 
\end{lemma}
\begin{proof}
Let $C_2\in\Theta(X,\Box Y)$. This implies $X\models_{\Box Y}C_2$ and there is no theory implicate $C$ of $X$ with respect to $\Box Y$ such that $C\models_{\Box Y}C_2$. Hence, $X\cup\Box Y\models C_2$ and there is no implicate $C$ of $X\cup\Box Y$ such that $C\models_{\Box Y}C_2$. But given that $C_1\in\Pi(X\cup\Box Y)$ and $C_1\models_{\Box Y}C_2$. So there is an implicate $C_1$ of $X\cup\Box Y$
 such that $C_1\models_{\Box Y}C_2$. This implies $C_2$ is not prime which is a contradiction. So, $C_2\not\in\Theta(X,\Box Y)$. \hfill{$\Box$}
\end{proof} 
So we conclude from Lemma \ref{tpi_pi_subset} and Lemma \ref{tpi_pi_notbelongto} that the set of theory prime implicates of $X$ with respect to $\Box Y$ can be defined from the set of prime implicates of $X\cup\Box Y$ as follows:

\begin{theorem}\label{minimal_wrt_pi}
$\Theta(X,\Box Y)=\mbox{min}(\Pi(X\cup\Box Y),\models_{\Box Y})$
\end{theorem}

The above theorem is used in proving the correctness of computation of theory prime implicate algorithm.

 The following theorem says that the set of theory prime implicates of $X$ with respect to $\Box Y$ captures all the theory implicates of $X$ with respect to $\Box Y$. It is useful in proving the correctness of query answering algorithm later.

\begin{theorem}\label{alt_def_of_tpi}
Let $X$ and $Y$ be modal formulas and $C$ be a clause. Then $X\models_{\Box Y}C$ holds if and only if there is a theory prime implicate $C^{'}$ of $X$ with respect to $\Box Y$ such that $C^{'}\models_{\Box Y}C$ holds.
\end{theorem}
\begin{proof}
Suppose $X\models_{\Box Y}C$, i.e, $X\cup\Box Y\models C$. So $C$ is an implicate of $X\cup\Box Y$. If $C$ is not prime then there is an implicate $C^{*}$ of $X\cup\Box Y$ such that $C^{*}\models C$. Let $A^{*}=\{C^{*}~|~ C^{*} \mbox{ is an implicate of } X\cup\Box Y \mbox{ and } C^{*}\models C\}$. We can find out a set $A=\{C_1,\ldots,C_n\}$ such that for each $C^{*}\in A^{*}$ there is a $C_i\in A$ such that $C_i\equiv_{\Box Y} C^{*}$. We have chosen one element per equivalence class. So each element of $A$ is an implicate of $X\cup\Box Y$. Then any minimal element of $A$ is a prime implicate of $X\cup\Box Y$. This means we obtain a theory prime implicate of $X$ with respect to $\Box Y$.

Conversely, there exists a theory prime implicate $C^{'}$ of $X$ with respect to $\Box Y$ such that $C^{'}\models_{\Box Y}C$ holds. This implies $X\cup\Box Y\models C^{'}$ and $C^{'}\cup\Box Y\models C$. As $C^{'}$ is disjunctive so $C^{'}\models C^{'}\cup\Box Y$. Hence $X\cup\Box Y\models C$.  \hfill{$\Box$}
\end{proof}

The following theorem is a metalogical property of prime implicates.

\begin{lemma}\label{plain_equivalence}
Let $X$ and $X^{'}$ be formulae in $\mathcal{T}$. Then $X\equiv X^{'}$ if and only if $\Pi(X)\equiv\Pi(X^{'})$.
\end{lemma}
\begin{proof}
It is easy to prove.
\end{proof}

The following theorem is a metalogical property of theory prime implicates.

\begin{theorem}\label{equivalence}
Suppose $X, X^{'},$ are formulae in $\mathcal{T}$ and $Y, Y^{'}$ be any propositional formulae. If $X\equiv_{\Box Y}X^{'}$ and $Y\equiv Y^{'}$ then $\Theta(X^{'},{\Box Y^{'}})=\Theta(X,\Box Y)$.
\end{theorem}
\begin{proof}
Let $C\in\Theta(X^{'},\Box Y^{'})$, i.e, $C$ is a theory prime implicate of $X^{'}$ with respect to $\Box Y^{'}$. So by definition $C$ is a theory implicate of $X^{'}$ with respect to $\Box Y^{'}$ and there is no other theory implicate $C^{'}$ of $X^{'}$ with respect to $\Box Y^{'}$ such that $C^{'}\models_{\Box Y}C$, i.e, $X^{'}\cup\Box Y^{'}\models C$ and there does not exist any $C^{'}$ such that $X^{'}\cup\Box Y^{'}\models C^{'}$ and $C^{'}\models_{\Box Y}C$. As $X\equiv_{\Box Y}X^{'}$ so $X\cup\Box Y\models X^{'}$. Hence $X\cup\Box Y\cup\Box Y^{'}\models X^{'}\cup\Box Y^{'}\models C$ so $X\cup\Box Y\cup\Box Y^{'}\models C$. As $Y\equiv Y^{'}$ so by Lemma \ref{DiamondBox_equiv}, we have $\Box Y\equiv\Box Y^{'}$. Hence $X\cup\Box Y\models C$. Hence, $X\cup\Box Y\models C$ and there does not exist any $C^{'}$ such that $X\cup\Box Y\models C^{'}$ and $C^{'}\models_{\Box Y}C$. This implies $C$ is a theory prime implicate of $X$ with respect to $\Box Y$, i.e, $C\in\Theta(X,\Box Y)$. This implies, $\Theta(X^{'},{\Box Y^{'}})\subseteq\Theta(X,\Box Y)$. Similarly we can prove $\Theta(X,\Box Y)\subseteq\Theta(X^{'},{\Box Y^{'}})$. Hence proved. \hfill{$\Box$}
\end{proof}

The equivalence preserving knowledge compilation Theorem \ref{tpi_compilation} will be proved by the help of the following lemma.

\begin{lemma}\label{equivalence_wrt_boxy}
Let $X$ be a modal formula and $Y$ be any propositional formula. Then $X\equiv_{\Box Y}\Theta(X,\Box Y)$. 
\end{lemma}
\begin{proof}
First we have to prove $X\models_{\Box Y}\Theta(X,\Box Y)$, i.e, to prove  $X\models_{\Box Y}\mbox{ min}(\Pi(X\cup\Box Y),\models_{\Box Y})$ (by Theorem \ref{minimal_wrt_pi}), i.e, to prove $X\cup \Box Y\models\mbox{ min}(\Pi(X\cup\Box Y),\models_{\Box Y})$.  

As $\mbox{nb\_cl}(\Pi(X\cup\Box Y))\geq \mbox{nb\_cl}(\mbox{min}(\Pi(X\cup\Box Y),\models_{\Box Y}))$ so  $\Pi(X\cup\Box Y)\models\mbox{ min}(\Pi(X\cup\Box Y),\models_{\Box Y})$. As $X\cup\Box Y\equiv\Pi(X\cup\Box Y)$, so $X\cup\Box Y\models\mbox{ min}(\Pi(X\cup\Box Y),\models_{\Box Y})$. 

Conversely, we have to prove $\Theta(X,Y)\models_{\Box Y}X$, i.e, to prove $\mbox{ min}(\Pi(X\cup\Box Y),\models_{\Box Y})\models_{\Box Y}X$ (by Theorem \ref{minimal_wrt_pi}), i.e, to prove $\mbox{ min}(\Pi(X\cup\Box Y),\models_{\Box Y})\cup\Box Y\models X$. As $\Pi(X\cup \Box Y)\cup \Box Y\models \mbox{ min}(\Pi(X\cup\Box Y),\models_{\Box Y})\cup\Box Y$, so we have to prove  $\Pi(X\cup \Box Y)\cup \Box Y\models X$. As $\Pi(X\cup \Box Y)\equiv X\cup \Box Y$ so we have to prove $(X\cup \Box Y)\cup \Box Y\models X$, i.e, to prove $X\cup \Box Y\models X$ which holds always. Hence proved. \hfill{$\Box$}
\end{proof}

The following result which holds in $\mathcal{T}$ shows that weakening the consequence $Y$ does not increases the number of clauses of $\Theta(X,\Box Y)$. 

\begin{theorem}\label{weaken_the_consequence}
Let $X$ be a modal formula and $Y, Y^{'}$ be any propositional formulas such that $X\models Y$ and $Y\models Y^{'}$. For every $\pi^{'}\in\Theta(X,\Box Y^{'})$ there exists a $\pi\in\Theta(X,\Box Y)$ such that $\pi^{'}\equiv_{\Box Y^{'}}\pi$. Consequently, $\mbox{nb\_cl}(\Theta(X,\Box Y^{'}))\leq\mbox{nb\_cl}(\Theta(X,\Box Y))$.
\end{theorem}

\begin{proof}
Let $X=(\vee_{i=1}^{l}\alpha_i)\wedge(\vee_{j=1}^{m}\Diamond\beta_j)\wedge(\vee_{k=1}^{n}\Box\gamma_{k})\wedge((\vee_{i=1}^{l}\alpha_i)\vee(\vee_{j=1}^{m}\Diamond\beta_j))\wedge((\vee_{i=1}^{l}\alpha_i)\vee(\vee_{k=1}^{n}\Box\gamma_{k}))\wedge((\vee_{j=1}^{m}\Diamond\beta_j)\vee(\vee_{k=1}^{n}\Box\gamma_{k}))\wedge((\vee_{i=1}^{l}\alpha_i)\vee(\vee_{j=1}^{m}\Diamond\beta_j)\vee(\vee_{k=1}^{n}\Box\gamma_{k}))$ be a formula in $\mathcal{K}$ where $\alpha_1,\alpha_2,\ldots,\alpha_l$ be propositional formulae and $\beta_1,\beta_2,\ldots,\beta_m,\gamma_1,\gamma_2,\ldots,\gamma_n$ be formulae in $\mathcal{K}$. Let $\pi^{'}=\psi_1^{'}\vee\ldots\vee\psi_p^{'}\vee\Diamond\phi_1^{'}\vee\ldots\vee\Diamond\phi_q^{'}\vee\Box\xi_1^{'}\vee\ldots\vee\Box\xi_r^{'}$ be a theory prime implicate of $X$ with respect to $\Box Y^{'}$ where $\psi_1,\psi_2,\ldots,\psi_p$ be propositional formulae and $\phi_1,\phi_2,\ldots,\phi_q,\xi_1,\xi_2,\ldots,\xi_r$ be formulae in $\mathcal{K}$. So $\pi^{'}\in\Theta(X,\Box Y^{'})$, i.e, $X\models_{\Box Y^{'}}\pi^{'}$, i.e, $X\wedge\neg\pi^{'}\models_{\Box Y^{'}}\bot$. This implies $(\vee_{i=1}^{l}\alpha_i)\wedge(\vee_{j=1}^{m}\Diamond\beta_j)\wedge(\vee_{k=1}^{n}\Box\gamma_{k})\wedge((\vee_{i=1}^{l}\alpha_i)\vee(\vee_{j=1}^{m}\Diamond\beta_j))\wedge((\vee_{i=1}^{l}\alpha_i)\vee(\vee_{k=1}^{n}\Box\gamma_{k}))\wedge((\vee_{j=1}^{m}\Diamond\beta_j)\vee(\vee_{k=1}^{n}\Box\gamma_{k}))\wedge((\vee_{i=1}^{l}\alpha_i)\vee(\vee_{j=1}^{m}\Diamond\beta_j)\vee(\vee_{k=1}^{n}\Box\gamma_{k}))\wedge\neg\psi_1^{'}\wedge\ldots\wedge\neg\psi_p^{'}\wedge\Box\neg\phi_1^{'}\wedge\ldots\wedge\Box\neg\phi_q^{'}\wedge\Diamond\neg\xi_1^{'}\wedge\ldots\wedge\Diamond\neg\xi_r^{'}\models_{\Box Y^{'}}\bot$. Then by Theorem \ref{several_possibilities_wrt_boxY}, one of the following will hold.
\begin{enumerate}
\item $(\vee_{i=1}^{l}\alpha_i)\wedge\neg\psi_1^{'}\wedge\ldots\wedge\neg\psi_p^{'}\models_{Y^{'}}\bot$. So $(\vee_{i=1}^{l}\alpha_i)\models_{Y^{'}}\psi_1^{'}\vee\ldots\vee\psi_p^{'}$. Hence by property of $\mathcal{T}$, $(\vee_{i=1}^{l}\alpha_i)\models_{\Box Y^{'}}\psi_1^{'}\vee\ldots\vee\psi_p^{'}$. Hence, $(\vee_{i=1}^{l}\alpha_i)\models_{\Box Y^{'}}\pi^{'}$. As $(\vee_{i=1}^{l}\alpha_i)$ is a theory implicate of $X$ with respect to $\Box Y^{'}$, so $(\vee_{i=1}^{l}\alpha_i)\equiv_{\Box Y^{'}}\pi^{'}$. As $Y\models Y^{'}$, so by Lemma \ref{DiamondBox_equiv}, $(\vee_{i=1}^{l}\alpha_i)\models_{\Box Y}(\vee_{i=1}^{l}\alpha_i)\models_{\Box Y^{'}}\pi^{'}$. So $(\vee_{i=1}^{l}\alpha_i)\models_{\Box Y}\pi^{'}$. So we note that $\pi^{'}$ is not a theory prime implicate of $X$ with respect to $\Box Y$ but as $X\models_{\Box Y}(\vee_{i=1}^{l}\alpha_i)$ so $(\vee_{i=1}^{l}\alpha_i)$ is a theory implicate of $X$ with respect to $\Box Y$. As there does not exist any $C$ such that $X\models_{\Box Y} C$ and $C\models_{\Box Y}(\vee_{i=1}^{l}\alpha_i)$ so $(\vee_{i=1}^{l}\alpha_i)$ is a theory prime implicate of $X$ with respect to $\Box Y$. Assuming $\pi=(\vee_{i=1}^{l}\alpha_i)$ we have shown that for every $\pi^{'}\in\Theta(X,\Box Y^{'})$ there exists a $\pi\in\Theta(X,\Box Y)$ such that $\pi^{'}\equiv_{\Box Y^{'}}\pi$. 
\item $(\vee_{j=1}^{m}\beta_j)\wedge\neg\phi_1^{'}\wedge\ldots\wedge\neg\phi_q^{'}\models_{Y^{'}}\bot$. So $(\vee_{j=1}^{m}\beta_j)\models_{Y^{'}}(\phi_1^{'}\vee\ldots\vee\phi_q^{'})$. Then by Lemma \ref{symmetry_entailment}, $\Diamond(\vee_{j=1}^{m}\beta_j)\models_{\Box Y^{'}}\Diamond(\phi_1^{'}\vee\ldots\vee\phi_q^{'})$. By property of $\mathcal{K}$, we have $(\vee_{j=1}^{m}\Diamond\beta_j)\models_{\Box Y^{'}}(\Diamond\phi_1^{'}\vee\ldots\vee\Diamond\phi_q^{'})$. So $(\vee_{j=1}^{m}\Diamond\beta_j)\models_{\Box Y^{'}}\pi^{'}$. As $(\vee_{j=1}^{m}\Diamond\beta_j)$ is a theory implicate of $X$ with respect to $\Box Y^{'}$ so $(\vee_{j=1}^{m}\Diamond\beta_j)\equiv_{\Box Y^{'}}\pi^{'}$. As $Y\models Y^{'}$, so by Lemma \ref{DiamondBox_equiv}, $(\vee_{j=1}^{m}\Diamond\beta_j)\models_{\Box Y}(\vee_{j=1}^{m}\Diamond\beta_j)\models_{\Box Y^{'}}\pi^{'}$. So $(\vee_{j=1}^{m}\Diamond\beta_j)\models_{\Box Y}\pi^{'}$. So $\pi^{'}$ is not a theory prime implicate of $X$ with respect to $\Box Y$ but as $X\models_{\Box Y}(\vee_{j=1}^{m}\Diamond\beta_j)$ so $(\vee_{j=1}^{m}\Diamond\beta_j)$ is a theory implicate of $X$ with respect to $\Box Y$. As there does not exist any $C$ such that $X\models_{\Box Y} C$ and $C\models_{\Box Y}(\vee_{j=1}^{m}\Diamond\beta_j)$ so $(\vee_{j=1}^{m}\Diamond\beta_j)$ is a theory prime implicate of $X$ with respect to $\Box Y$. Assuming $(\vee_{j=1}^{m}\Diamond\beta_j)=\pi$, we have shown that for every $\pi^{'}\in\Theta(X,\Box Y^{'})$ there exists a $\pi\in\Theta(X,\Box Y)$ such that $\pi^{'}\equiv_{\Box Y^{'}}\pi$. 
\item $(\vee_{k=1}^{n}\gamma_{k})\wedge\neg\xi_u^{'}\wedge\neg\phi_1^{'}\wedge\ldots\wedge\phi_q^{'}\models_{Y^{'}}\bot$ for $1\leq u\leq r$. So $(\vee_{k=1}^{n}\gamma_{k})\models_{Y^{'}}\xi_u^{'}\vee\phi_1^{'}\vee\ldots\vee\phi_q^{'}$. so by Lemma \ref{symmetry_entailment}, $\Box(\vee_{k=1}^{n}\gamma_{k})\models_{\Box Y^{'}}\Box(\xi_u^{'}\vee\phi_1^{'}\vee\ldots\vee\phi_q^{'})$. By property of $\mathcal{K}$ we have, $(\vee_{k=1}^{n}\Box\gamma_{k})\models_{\Box Y^{'}}\Box(\vee_{k=1}^{n}\gamma_{k})\models_{\Box Y^{'}}\Box\xi_u^{'}\vee\Diamond(\phi_1^{'}\vee\ldots\vee\phi_q^{'})\models_{\Box Y^{'}}\Box\xi_u^{'}\vee\Diamond\phi_1^{'}\vee\ldots\vee\Diamond\phi_q^{'}$. So, $(\vee_{k=1}^{n}\Box\gamma_{k})\models_{\Box Y^{'}}\Box\xi_u^{'}\vee\Diamond\phi_1^{'}\vee\ldots\vee\Diamond\phi_q^{'}$. So $(\vee_{k=1}^{n}\Box\gamma_{k})\models_{\Box Y^{'}}\pi^{'}$. As $(\vee_{k=1}^{n}\Box\gamma_{k})$ is a theory implicate of $X$ with respect to $\Box Y^{'}$, So $(\vee_{k=1}^{n}\Box\gamma_{k})\equiv_{\Box Y^{'}}\pi^{'}$. As $Y\models Y^{'}$, so by Lemma \ref{DiamondBox_equiv}, $(\vee_{k=1}^{n}\Box\gamma_{k})\models_{\Box Y}(\vee_{k=1}^{n}\Box\gamma_{k})\models_{\Box Y^{'}}\pi^{'}$ so $(\vee_{k=1}^{n}\Box\gamma_{k})\models_{\Box Y}\pi^{'}$. So $\pi^{'}$ is not a theory prime implicate of $X$ with respect to $\Box Y$ but as $X\models_{\Box Y}(\vee_{k=1}^{n}\Box\gamma_{k})$, so $(\vee_{k=1}^{n}\Box\gamma_{k})$ is a theory implicate of $X$ with respect to $\Box Y$. As there does not exist any $C$ such that $X\models_{\Box Y} C$ and $C\models_{\Box Y}(\vee_{k=1}^{n}\Box\gamma_{k})$ so $(\vee_{k=1}^{n}\Box\gamma_{k})$ is theory prime implicate of $X$ with respect to $\Box Y$. Assuming $(\vee_{k=1}^{n}\Box\gamma_{k})=\pi$, we have shown that for every $\pi^{'}\in\Theta(X,\Box Y^{'})$ there exists a $\pi\in\Theta(X,\Box Y)$ such that $\pi^{'}\equiv_{\Box Y^{'}}\pi$. 
\item $((\vee_{i=1}^{l}\alpha_i)\vee(\vee_{j=1}^{m}\beta_j))\wedge\neg\phi_1^{'}\wedge\ldots\wedge\neg\phi_q^{'}\models_{Y^{'}}\bot$. So, $((\vee_{i=1}^{l}\alpha_i)\vee(\vee_{j=1}^{m}\beta_j))\models_{Y^{'}}\phi_1^{'}\vee\ldots\vee\phi_q^{'}$. By Lemma \ref{symmetry_entailment}, $\Diamond((\vee_{i=1}^{l}\alpha_i)\vee(\vee_{j=1}^{m}\beta_j))\models_{\Box Y^{'}}\Diamond(\phi_1^{'}\vee\ldots\vee\phi_q^{'})$. By property of $\mathcal{K}$, $((\vee_{i=1}^{l}\Diamond\alpha_i)\vee(\vee_{j=1}^{m}\Diamond\beta_j))\models_{\Box Y^{'}}\Diamond\phi_1^{'}\vee\ldots\vee\Diamond\phi_q^{'}\models_{\Box Y^{'}}\pi^{'}$. Again by property of $\mathcal{T}$, $((\vee_{i=1}^{l}\alpha_i)\vee(\vee_{j=1}^{m}\Diamond\beta_j))\models_{\Box Y^{'}}((\vee_{i=1}^{l}\Diamond\alpha_i)\vee(\vee_{j=1}^{m}\Diamond\beta_j))\models_{\Box Y^{'}}\pi^{'}$. As $((\vee_{i=1}^{l}\alpha_i)\vee(\vee_{j=1}^{m}\Diamond\beta_j))$ is a theory implicate of $X$ with respect to $\Box Y^{'}$, so $((\vee_{i=1}^{l}\alpha_i)\vee(\vee_{j=1}^{m}\Diamond\beta_j))\equiv_{\Box Y^{'}}\pi^{'}$. As $Y\models Y^{'}$, so by Lemma \ref{DiamondBox_equiv}, $((\vee_{i=1}^{l}\alpha_i)\vee(\vee_{j=1}^{m}\Diamond\beta_j))\models_{\Box Y}((\vee_{i=1}^{l}\alpha_i)\vee(\vee_{j=1}^{m}\Diamond\beta_j))\models_{\Box Y^{'}}\pi^{'}$. So $((\vee_{i=1}^{l}\alpha_i)\vee(\vee_{j=1}^{m}\Diamond\beta_j))\models_{\Box Y}\pi^{'}$. So $\pi^{'}$ is not a theory prime implicate of $X$ with respect to $\Box Y$ but as $X\models_{\Box Y}((\vee_{i=1}^{l}\alpha_i)\vee(\vee_{j=1}^{m}\Diamond\beta_j))$ so $((\vee_{i=1}^{l}\alpha_i)\vee(\vee_{j=1}^{m}\Diamond\beta_j))$ is a theory implicate of $X$ with respect to $\Box Y$.
As there does not exist any $C$ such that $X\models_{\Box Y} C$ and $C\models_{\Box Y}((\vee_{i=1}^{l}\alpha_i)\vee(\vee_{j=1}^{m}\Diamond\beta_j))$, so $((\vee_{i=1}^{l}\alpha_i)\vee(\vee_{j=1}^{m}\Diamond\beta_j))$ is a theory prime implicate of $X$ with respect to $\Box Y$. Assuming $\pi=((\vee_{i=1}^{l}\alpha_i)\vee(\vee_{j=1}^{m}\Diamond\beta_j))$ we have shown that for every $\pi^{'}\in\Theta(X,\Box Y^{'})$ there exists a $\pi\in\Theta(X,\Box Y)$ such that $\pi^{'}\equiv_{\Box Y^{'}}\pi$.
\item $((\vee_{i=1}^{l}\alpha_i)\vee(\vee_{k=1}^{n}\gamma_{k}))\wedge\neg\xi_u^{'}\wedge\neg\phi_1^{'}\wedge\ldots\wedge\neg\phi_q^{'}\models_{Y^{'}}\bot$ for $1\leq u\leq r$. This implies, $((\vee_{i=1}^{l}\alpha_i)\vee(\vee_{k=1}^{n}\gamma_{k}))\models_{Y^{'}}\xi_u^{'}\vee\phi_1^{'}\vee\ldots\vee\phi_q^{'}$. By Lemma \ref{symmetry_entailment}, $\Box((\vee_{i=1}^{l}\alpha_i)\vee(\vee_{k=1}^{n}\gamma_{k}))\models_{\Box Y^{'}}\Box(\xi_u^{'}\vee\phi_1^{'}\vee\ldots\vee\phi_q^{'})$. By property of $\mathcal{K}$, $((\vee_{i=1}^{l}\Box\alpha_i)\vee(\vee_{k=1}^{n}\Box\gamma_{k}))\models_{\Box Y^{'}} \Box((\vee_{i=1}^{l}\alpha_i)\vee(\vee_{k=1}^{n}\gamma_{k}))$ and $\Box(\xi_u^{'}\vee\phi_1^{'}\vee\ldots\vee\phi_q^{'})\models_{\Box Y^{'}}\Box\xi_u^{'}\vee\Diamond(\phi_1^{'}\vee\ldots\vee\phi_q^{'}))\models_{\Box Y^{'}}\Box\xi_u^{'}\vee\Diamond\phi_1^{'}\vee\ldots\vee\Diamond\phi_q^{'})\models_{\Box Y^{'}}\pi^{'}$, so $((\vee_{i=1}^{l}\Box\alpha_i)\vee(\vee_{k=1}^{n}\Box\gamma_{k}))\models_{\Box Y^{'}}\pi^{'}$. As $\alpha_i$ for $1\leq i\leq l$ is a propositional formula, $((\vee_{i=1}^{l}\alpha_i)\vee(\vee_{k=1}^{n}\Box\gamma_{k}))\models_{\Box Y^{'}}((\vee_{i=1}^{l}\Box\alpha_i)\vee(\vee_{k=1}^{n}\Box\gamma_{k}))\models_{\Box Y^{'}}\pi^{'}$. As $((\vee_{i=1}^{l}\alpha_i)\vee(\vee_{k=1}^{n}\Box\gamma_{k}))$ is a theory implicate of $X$ with respect to $\Box Y^{'}$, so $((\vee_{i=1}^{l}\alpha_i)\vee(\vee_{k=1}^{n}\Box\gamma_{k}))\equiv_{\Box Y^{'}}\pi^{'}$. Again as $Y\models Y^{'}$, so by Lemma \ref{DiamondBox_equiv}, $((\vee_{i=1}^{l}\alpha_i)\vee(\vee_{k=1}^{n}\Box\gamma_{k}))\models_{\Box Y}((\vee_{i=1}^{l}\alpha_i)\vee(\vee_{k=1}^{n}\Box\gamma_{k}))\models_{\Box Y^{'}}\pi^{'}$. Hence $((\vee_{i=1}^{l}\alpha_i)\vee(\vee_{k=1}^{n}\Box\gamma_{k}))\models_{\Box Y}\pi^{'}$. So, $\pi^{'}$ is not a theory prime implicate of $X$ with respect to $\Box Y$ but as $X\models_{\Box Y}((\vee_{i=1}^{l}\alpha_i)\vee(\vee_{k=1}^{n}\Box\gamma_{k}))$, so $((\vee_{i=1}^{l}\alpha_i)\vee(\vee_{k=1}^{n}\Box\gamma_{k}))$ is a theory implicate of $X$ with respect to $\Box Y$. As there does not exist any $C$ such that $X\models_{\Box Y} C$ and $C\models_{\Box Y}((\vee_{i=1}^{l}\alpha_i)\vee(\vee_{k=1}^{n}\Box\gamma_{k}))$, So $((\vee_{i=1}^{l}\alpha_i)\vee(\vee_{k=1}^{n}\Box\gamma_{k}))$ is a theory prime implicate of $X$ with respect to $\Box Y$. Assuming $\pi=((\vee_{i=1}^{l}\alpha_i)\vee(\vee_{k=1}^{n}\Box\gamma_{k}))$ we have shown that for every $\pi^{'}\in\Theta(X,\Box Y^{'})$ there exists a $\pi\in\Theta(X,\Box Y)$ such that $\pi^{'}\equiv_{\Box Y^{'}}\pi$.
\item $((\vee_{j=1}^{m}\beta_j)\vee(\vee_{k=1}^{n}\gamma_{k}))\wedge\neg\phi_1^{'}\wedge\ldots\wedge\neg\phi_q^{'}\models_{Y^{'}}\bot$. This implies, $((\vee_{j=1}^{m}\beta_j)\vee(\vee_{k=1}^{n}\gamma_{k}))\models_{Y^{'}}\phi_1^{'}\vee\ldots\vee\phi_q^{'}$. By Lemma \ref{symmetry_entailment}, we have $\Diamond((\vee_{j=1}^{m}\beta_j)\vee(\vee_{k=1}^{n}\gamma_{k}))\models_{\Box Y^{'}}\Diamond(\phi_1^{'}\vee\ldots\vee\phi_q^{'})$. By property of $\mathcal{K}$, $((\vee_{j=1}^{m}\Diamond\beta_j)\vee(\vee_{k=1}^{n}\Diamond\gamma_{k}))\models_{\Box Y^{'}}(\Diamond\phi_1^{'}\vee\ldots\vee\Diamond\phi_q^{'})$. Again by property of $\mathcal{T}$, $((\vee_{j=1}^{m}\Diamond\beta_j)\vee(\vee_{k=1}^{n}\Box\gamma_{k}))\models_{\Box Y^{'}}((\vee_{j=1}^{m}\Diamond\beta_j)\vee(\vee_{k=1}^{n}\Diamond\gamma_{k}))\models_{\Box Y^{'}}(\Diamond\phi_1^{'}\vee\ldots\vee\Diamond\phi_q^{'})\models_{\Box Y^{'}}\pi^{'}$. So $((\vee_{j=1}^{m}\Diamond\beta_j)\vee(\vee_{k=1}^{n}\Box\gamma_{k}))\models_{\Box Y^{'}}\pi^{'}$. As $((\vee_{j=1}^{m}\Diamond\beta_j)\vee(\vee_{k=1}^{n}\Box\gamma_{k}))$ is a theory implicate of $X$ with respect to $\Box Y^{'}$, so $((\vee_{j=1}^{m}\Diamond\beta_j)\vee(\vee_{k=1}^{n}\Box\gamma_{k}))\equiv_{\Box Y^{'}}\pi^{'}$. As $Y\models Y^{'}$, so by Lemma \ref{DiamondBox_equiv} $((\vee_{j=1}^{m}\Diamond\beta_j)\vee(\vee_{k=1}^{n}\Box\gamma_{k}))\models_{\Box Y}((\vee_{j=1}^{m}\Diamond\beta_j)\vee(\vee_{k=1}^{n}\Box\gamma_{k}))\models_{\Box Y^{'}}\pi^{'}$. Hence, $((\vee_{j=1}^{m}\Diamond\beta_j)\vee(\vee_{k=1}^{n}\Box\gamma_{k}))\models_{\Box Y}\pi^{'}$. So $\pi^{'}$ is not a theory prime implicate of $X$ with respect to $\Box Y$ but as $X\models_{\Box Y}((\vee_{j=1}^{m}\Diamond\beta_j)\vee(\vee_{k=1}^{n}\Box\gamma_{k}))$ so, $((\vee_{j=1}^{m}\Diamond\beta_j)\vee\vee_{k=1}^{n}\Box\gamma_{k}))$ is a theory implicate of $X$ with respect to $\Box Y$. As there does not exist any $C$ such that $X\models_{\Box Y} C$ and $C\models_{\Box Y}((\vee_{j=1}^{m}\Diamond\beta_j)\vee(\vee_{k=1}^{n}\Box\gamma_{k}))$, so $((\vee_{j=1}^{m}\Diamond\beta_j)\vee(\vee_{k=1}^{n}\Box\gamma_{k}))$ is a theory prime implicate of $X$ with respect to $\Box Y$. Assuming $((\vee_{j=1}^{m}\Diamond\beta_j)\vee(\vee_{k=1}^{n}\Box\gamma_{k}))=\pi$, we have shown that for every $\pi^{'}\in\Theta(X,\Box Y^{'})$ there exists a $\pi\in\Theta(X,\Box Y)$ such that $\pi^{'}\equiv_{\Box Y^{'}}\pi$.
\item $((\vee_{i=1}^{l}\alpha_i)\vee(\vee_{j=1}^{m}\beta_j)\vee(\vee_{k=1}^{n}\gamma_{k}))\wedge\neg\phi_1^{'}\wedge\ldots\wedge\neg\phi_q^{'}\models_{Y^{'}}\bot$. This implies, $((\vee_{i=1}^{l}\alpha_i)\vee(\vee_{j=1}^{m}\beta_j)\vee(\vee_{k=1}^{n}\gamma_{k}))\models_{Y^{'}}\phi_1^{'}\vee\ldots\vee\phi_q^{'}$. By Lemma \ref{symmetry_entailment}, $\Diamond((\vee_{i=1}^{l}\alpha_i)\vee(\vee_{j=1}^{m}\beta_j)\vee(\vee_{k=1}^{n}\gamma_{k}))\models_{\Box Y^{'}}\Diamond(\phi_1^{'}\vee\ldots\vee\phi_q^{'})$. By property of $\mathcal{K}$, $((\vee_{i=1}^{l}\Diamond\alpha_i)\vee(\vee_{j=1}^{m}\Diamond\beta_j)\vee(\vee_{k=1}^{n}\Diamond\gamma_{k}))\models_{\Box Y^{'}}(\Diamond\phi_1^{'}\vee\ldots\vee\Diamond\phi_q^{'})$. Again by property of $\mathcal{T}$, $((\vee_{i=1}^{l}\alpha_i)\vee(\vee_{j=1}^{m}\Diamond\beta_j)\vee(\vee_{k=1}^{n}\Box\gamma_{k}))\models_{\Box Y^{'}}((\vee_{i=1}^{l}\Diamond\alpha_i)\vee(\vee_{j=1}^{m}\Diamond\beta_j)\vee(\vee_{k=1}^{n}\Diamond\gamma_{k}))\models_{\Box Y^{'}}(\Diamond\phi_1^{'}\vee\ldots\vee\Diamond\phi_q^{'})\models_{\Box Y^{'}}\pi^{'}$. So, $((\vee_{i=1}^{l}\alpha_i)\vee(\vee_{j=1}^{m}\Diamond\beta_j)\vee(\vee_{k=1}^{n}\Box\gamma_{k}))\models_{\Box Y^{'}}\pi^{'}$. As $((\vee_{i=1}^{l}\alpha_i)\vee(\vee_{j=1}^{m}\Diamond\beta_j)\vee(\vee_{k=1}^{n}\Box\gamma_{k}))$ is a theory implicate of $X$ with respect to $\Box Y^{'}$ so, $((\vee_{i=1}^{l}\alpha_i)\vee(\vee_{j=1}^{m}\Diamond\beta_j)\vee(\vee_{k=1}^{n}\Box\gamma_{k}))\equiv_{\Box Y^{'}}\pi^{'}$. As $Y\models Y^{'}$, so by Lemma \ref{DiamondBox_equiv}, $((\vee_{i=1}^{l}\alpha_i)\vee(\vee_{j=1}^{m}\Diamond\beta_j)\vee(\vee_{k=1}^{n}\Box\gamma_{k}))\models_{\Box Y}((\vee_{i=1}^{l}\alpha_i)\vee(\vee_{j=1}^{m}\Diamond\beta_j)\vee(\vee_{k=1}^{n}\Box\gamma_{k}))\models_{\Box Y^{'}}\pi^{'}$. Hence, $((\vee_{i=1}^{l}\alpha_i)\vee(\vee_{j=1}^{m}\Diamond\beta_j)\vee(\vee_{k=1}^{n}\Box\gamma_{k}))\models_{\Box Y}\pi^{'}$. So $\pi^{'}$ is not a theory prime implicate of $X$ with respect to $\Box Y$ but as $X\models_{\Box Y}((\vee_{i=1}^{l}\alpha_i)\vee(\vee_{j=1}^{m}\Diamond\beta_j)\vee(\vee_{k=1}^{n}\Box\gamma_{k}))$ so  $((\vee_{i=1}^{l}\alpha_i)\vee(\vee_{j=1}^{m}\Diamond\beta_j)\vee(\vee_{k=1}^{n}\Box\gamma_{k}))$ is theory implicate of $X$ with respect to $\Box Y$. As there does not exist any $C$ such that $X\models_{\Box Y} C$ and $C\models_{\Box Y}((\vee_{i=1}^{l}\alpha_i)\vee(\vee_{j=1}^{m}\Diamond\beta_j)\vee(\vee_{k=1}^{n}\Box\gamma_{k}))$ so $((\vee_{i=1}^{l}\alpha_i)\vee(\vee_{j=1}^{m}\Diamond\beta_j)\vee(\vee_{k=1}^{n}\Box\gamma_{k}))$ is a theory prime implicate of $X$ with respect to $\Box Y$. Assuming $((\vee_{i=1}^{l}\alpha_i)\vee(\vee_{j=1}^{m}\Diamond\beta_j)\vee(\vee_{k=1}^{n}\Box\gamma_{k}))=\pi$,  we have shown that for every $\pi^{'}\in\Theta(X,\Box Y^{'})$ there exists a $\pi\in\Theta(X,\Box Y)$ such that $\pi^{'}\equiv_{\Box Y^{'}}\pi$. Hence proved. \hfill{$\Box$}
\end{enumerate}
\end{proof}

The following result which holds in $\mathcal{T}$ shows that the size of $\Theta(X,\Box Y)$ is always smaller than the size of $\Pi(X\cup\Box Y)$ which is an advantage to our compilation.

\begin{theorem}\label{tpi_size_less_than_pi}
Let $X$ be a modal formula and $Y$ be any propositional formula. For every $\pi^{'}\in\Theta(X, \Box Y)$ there exists a $\pi\in\Pi(X\cup\Box Y)$ such that $\pi^{'}\equiv_{\Box Y}\pi$. Consequently, $\mbox{nb\_cl}(\Theta(X,\Box Y)) \leq\mbox{nb\_cl}(\Pi(X\cup \Box Y))$.
\end{theorem}
\begin{proof}
Let $X=(\vee_{i=1}^{l}\alpha_i)\wedge(\vee_{j=1}^{m}\Diamond\beta_j)\wedge(\vee_{k=1}^{n}\Box\gamma_{k})\wedge((\vee_{i=1}^{l}\alpha_i)\vee(\vee_{j=1}^{m}\Diamond\beta_j))\wedge((\vee_{i=1}^{l}\alpha_i)\vee(\vee_{k=1}^{n}\Box\gamma_{k}))\wedge((\vee_{j=1}^{m}\Diamond\beta_j)\vee(\vee_{k=1}^{n}\Box\gamma_{k}))\wedge((\vee_{i=1}^{l}\alpha_i)\vee(\vee_{j=1}^{m}\Diamond\beta_j)\vee(\vee_{k=1}^{n}\Box\gamma_{k}))$ be a formula in $\mathcal{K}$ where $\alpha_1,\alpha_2,\ldots,\alpha_l$ be propositional formulae and $\beta_1,\beta_2,\ldots,\beta_m,\gamma_1,\gamma_2,\ldots,\gamma_n$ be formulae in $\mathcal{K}$. Let $\pi^{'}=\psi_1^{'}\vee\ldots\vee\psi_p^{'}\vee\Diamond\phi_1^{'}\vee\ldots\vee\Diamond\phi_q^{'}\vee\Box\xi_1^{'}\vee\ldots\vee\Box\xi_r^{'}$ be a theory prime implicate of $X$ with respect to $\Box Y$, where $\psi_1,\psi_2,\ldots,\psi_p$ be propositional formulae and $\phi_1,\phi_2,\ldots,\phi_q,\xi_1,\xi_2,\ldots,\xi_r$ be formulae in $\mathcal{K}$. So $\pi^{'}\in\Theta(X,\Box Y)$, i.e, $X\models_{\Box Y}\pi^{'}$, i.e, $X\wedge\neg\pi^{'}\models_{\Box Y}\bot$. This implies $(\vee_{i=1}^{l}\alpha_i)\wedge(\vee_{j=1}^{m}\Diamond\beta_j)\wedge(\vee_{k=1}^{n}\Box\gamma_{k})\wedge((\vee_{i=1}^{l}\alpha_i)\vee(\vee_{j=1}^{m}\Diamond\beta_j))\wedge((\vee_{i=1}^{l}\alpha_i)\vee(\vee_{k=1}^{n}\Box\gamma_{k}))\wedge((\vee_{j=1}^{m}\Diamond\beta_j)\vee(\vee_{k=1}^{n}\Box\gamma_{k}))\wedge((\vee_{i=1}^{l}\alpha_i)\vee(\vee_{j=1}^{m}\Diamond\beta_j)\vee(\vee_{k=1}^{n}\Box\gamma_{k}))\wedge\neg\psi_1^{'}\wedge\ldots\wedge\neg\psi_p^{'}\wedge\Box\neg\phi_1^{'}\wedge\ldots\wedge\Box\neg\phi_q^{'}\wedge\Diamond\neg\xi_1^{'}\wedge\ldots\wedge\Diamond\neg\xi_r^{'}\models_{\Box Y}\bot$. Then by Theorem \ref{several_possibilities_wrt_boxY}, one of the following will hold.
\begin{enumerate}
\item $(\vee_{i=1}^{l}\alpha_i)\wedge\neg\psi_1^{'}\wedge\ldots\wedge\neg\psi_p^{'}\models_{Y}\bot$. So $(\vee_{i=1}^{l}\alpha_i)\models_{Y}\psi_1^{'}\vee\ldots\vee\psi_p^{'}\models_{Y}\pi^{'}$. Hence by property of $\mathcal{T}$, $(\vee_{i=1}^{l}\alpha_i)\models_{\Box Y}\pi^{'}$. As $(\vee_{i=1}^{l}\alpha_i)$ is a theory implicate of $X$ with respect to $\Box Y$, so $(\vee_{i=1}^{l}\alpha_i)\equiv_{\Box Y}\pi^{'}$. As $X\wedge\Box Y\models\vee_{i=1}^{l}\alpha_i$ so $\vee_{i=1}^{l}\alpha_i$ is an implicate of  $X\wedge \Box Y$. As there doesn't exist any $C$ such that $X\wedge \Box Y\models C$ and $C\models\vee_{i=1}^{l}\alpha_i$ so $\vee_{i=1}^{l}\alpha_i$ is a prime implicate of $X\wedge\Box Y$. Assuming $\pi=\vee_{i=1}^{l}\alpha_i$ we have for every $\pi^{'}\in\Theta(X, \Box Y)$ there exists a $\pi\in\Pi(X\cup\Box Y)$ such that $\pi^{'}\equiv_{\Box Y}\pi$.  
\item $(\vee_{j=1}^{m}\beta_j)\wedge\neg\phi_1^{'}\wedge\ldots\wedge\neg\phi_q^{'}\models_{Y}\bot$. So $(\vee_{j=1}^{m}\beta_j)\models_{Y}(\phi_1^{'}\vee\ldots\vee\phi_q^{'})$. Then by Lemma \ref{symmetry_entailment}, $\Diamond(\vee_{j=1}^{m}\beta_j)\models_{\Box Y}\Diamond(\phi_1^{'}\vee\ldots\vee\phi_q^{'})$. By property of $\mathcal{K}$, we have $(\vee_{j=1}^{m}\Diamond\beta_j)\models_{\Box Y}(\Diamond\phi_1^{'}\vee\ldots\vee\Diamond\phi_q^{'})$. So $(\vee_{j=1}^{m}\Diamond\beta_j)\models_{\Box Y}\pi^{'}$. As $(\vee_{j=1}^{m}\Diamond\beta_j)$ is a theory implicate of $X$ with respect to $\Box Y$ so $(\vee_{j=1}^{m}\Diamond\beta_j)\equiv_{\Box Y}\pi^{'}$. As $X\wedge\Box Y\models\vee_{j=1}^{m}\beta_j$, so $\vee_{j=1}^{m}\beta_j$ is an implicate of $X\wedge\Box Y$.  As there doesn't exist any $C$ such that $X\wedge\Box Y\models C$ and $C\models\vee_{j=1}^{m}\beta_j$, so $\vee_{j=1}^{m}\beta_j$ is a prime implicate of $X\wedge\Box Y$. Assuming $\pi=\vee_{j=1}^{m}\beta_j$  we have for every $\pi^{'}\in\Theta(X, \Box Y)$ there exists a $\pi\in\Pi(X\cup\Box Y)$ such that $\pi^{'}\equiv_{\Box Y}\pi$.
\item $(\vee_{k=1}^{n}\gamma_{k})\wedge\neg\xi_u^{'}\wedge\neg\phi_1^{'}\wedge\ldots\wedge\phi_q^{'}\models_{Y}\bot$ for $1\leq u\leq r$. So $(\vee_{k=1}^{n}\gamma_{k})\models_{Y}\xi_u^{'}\vee\phi_1^{'}\vee\ldots\vee\phi_q^{'}$. so by Lemma \ref{symmetry_entailment}, $\Box(\vee_{k=1}^{n}\gamma_{k})\models_{\Box Y}\Box(\xi_u^{'}\vee\phi_1^{'}\vee\ldots\vee\phi_q^{'})$. By property of $\mathcal{K}$, $(\vee_{k=1}^{n}\Box\gamma_{k})\models_{\Box Y}\Box(\vee_{k=1}^{n}\gamma_{k})\models_{\Box Y}\Box\xi_u^{'}\vee\Box(\phi_1^{'}\vee\ldots\vee\phi_q^{'})\models_{\Box Y}\Box\xi_u^{'}\vee\Diamond\phi_1^{'}\vee\ldots\vee\Diamond\phi_q^{'}$. So, $(\vee_{k=1}^{n}\Box\gamma_{k})\models_{\Box Y}\Box\xi_u^{'}\vee\Diamond\phi_1^{'}\vee\ldots\vee\Diamond\phi_q^{'}$. So $(\vee_{k=1}^{n}\Box\gamma_{k})\models_{\Box Y}\pi^{'}$. As $(\vee_{k=1}^{n}\Box\gamma_{k})$ is a theory implicate of $X$ with respect to $\Box Y$, So $(\vee_{k=1}^{n}\Box\gamma_{k})\equiv_{\Box Y}\pi^{'}$. As $X\wedge\Box Y\models\vee_{k=1}^{n}\Box\gamma_{k}$, so $\vee_{k=1}^{n}\Box\gamma_{k}$ is an implicate of $X\wedge\Box Y$.  As there doesn't exist any $C$ such that $X\wedge\Box Y\models C$ and $C\models\vee_{k=1}^{n}\Box\gamma_{k}$, so $\vee_{k=1}^{n}\Box\gamma_{k}$ is a prime implicate of $X\wedge\Box Y$. Assuming $\pi=\vee_{k=1}^{n}\Box\gamma_{k}$  we have for every $\pi^{'}\in\Theta(X, \Box Y)$ there exists a $\pi\in\Pi(X\cup\Box Y)$ such that $\pi^{'}\equiv_{\Box Y}\pi$.
\item $((\vee_{i=1}^{l}\alpha_i)\vee(\vee_{j=1}^{m}\beta_j))\wedge\neg\phi_1^{'}\wedge\ldots\wedge\neg\phi_q^{'}\models_{Y}\bot$. So, $((\vee_{i=1}^{l}\alpha_i)\vee(\vee_{j=1}^{m}\beta_j))\models_{Y}\phi_1^{'}\vee\ldots\vee\phi_q^{'}$. By Lemma \ref{symmetry_entailment}, $\Diamond((\vee_{i=1}^{l}\alpha_i)\vee(\vee_{j=1}^{m}\beta_j))\models_{\Box Y}\Diamond(\phi_1^{'}\vee\ldots\vee\phi_q^{'})$. By property of $\mathcal{K}$, $((\vee_{i=1}^{l}\Diamond\alpha_i)\vee(\vee_{j=1}^{m}\Diamond\beta_j))\models_{\Box Y}\Diamond\phi_1^{'}\vee\ldots\vee\Diamond\phi_q^{'}\models_{\Box Y}\pi^{'}$. Again by property of $\mathcal{T}$, $((\vee_{i=1}^{l}\alpha_i)\vee(\vee_{j=1}^{m}\Diamond\beta_j))\models_{\Box Y}((\vee_{i=1}^{l}\Diamond\alpha_i)\vee(\vee_{j=1}^{m}\Diamond\beta_j))\models_{\Box Y}\pi^{'}$. As $((\vee_{i=1}^{l}\alpha_i)\vee(\vee_{j=1}^{m}\Diamond\beta_j))$ is a theory implicate of $X$ with respect to $\Box Y$, so $((\vee_{i=1}^{l}\alpha_i)\vee(\vee_{j=1}^{m}\Diamond\beta_j))\equiv_{\Box Y}\pi^{'}$. As $X\wedge\Box Y\models(\vee_{i=1}^{l}\alpha_i)\vee(\vee_{j=1}^{m}\Diamond\beta_j)$, so $(\vee_{i=1}^{l}\alpha_i)\vee(\vee_{j=1}^{m}\Diamond\beta_j)$ is an implicate of $X\wedge\Box Y$.  As there doesn't exist any $C$ such that $X\wedge\Box Y\models C$ and $C\models(\vee_{i=1}^{l}\alpha_i)\vee(\vee_{j=1}^{m}\Diamond\beta_j)$, so $(\vee_{i=1}^{l}\alpha_i)\vee(\vee_{j=1}^{m}\Diamond\beta_j)$ is a prime implicate of $X\wedge\Box Y$. Assuming $\pi=(\vee_{i=1}^{l}\alpha_i)\vee(\vee_{j=1}^{m}\Diamond\beta_j)$  we have for every $\pi^{'}\in\Theta(X, \Box Y)$ there exists a $\pi\in\Pi(X\cup\Box Y)$ such that $\pi^{'}\equiv_{\Box Y}\pi$.
\item $((\vee_{i=1}^{l}\alpha_i)\vee(\vee_{k=1}^{n}\gamma_{k}))\wedge\neg\xi_u^{'}\wedge\neg\phi_1^{'}\wedge\ldots\wedge\neg\phi_q^{'}\models_{Y}\bot$ for $1\leq u\leq r$. This implies, $((\vee_{i=1}^{l}\alpha_i)\vee(\vee_{k=1}^{n}\gamma_{k}))\models_{Y}\xi_u^{'}\vee\phi_1^{'}\vee\ldots\vee\phi_q^{'}$. By Lemma \ref{symmetry_entailment}, $\Box((\vee_{i=1}^{l}\alpha_i)\vee(\vee_{k=1}^{n}\gamma_{k}))\models_{\Box Y}\Box(\xi_u^{'}\vee\phi_1^{'}\vee\ldots\vee\phi_q^{'})$. By property of $\mathcal{K}$, $((\vee_{i=1}^{l}\Box\alpha_i)\vee(\vee_{k=1}^{n}\Box\gamma_{k}))\models_{\Box Y} \Box((\vee_{i=1}^{l}\alpha_i)\vee(\vee_{k=1}^{n}\gamma_{k}))$ and $\Box(\xi_u^{'}\vee\phi_1^{'}\vee\ldots\vee\phi_q^{'})\models_{\Box Y}\Box\xi_u^{'}\vee\Diamond(\phi_1^{'}\vee\ldots\vee\phi_q^{'}))\models_{\Box Y}\Box\xi_u^{'}\vee\Diamond\phi_1^{'}\vee\ldots\vee\Diamond\phi_q^{'})\models_{\Box Y}\pi^{'}$, so $((\vee_{i=1}^{l}\Box\alpha_i)\vee(\vee_{k=1}^{n}\Box\gamma_{k}))\models_{\Box Y}\pi^{'}$. As $\alpha_i$ for $1\leq i\leq l$ is a propositional formula, $((\vee_{i=1}^{l}\alpha_i)\vee(\vee_{k=1}^{n}\Box\gamma_{k}))\models_{\Box Y}((\vee_{i=1}^{l}\Box\alpha_i)\vee(\vee_{k=1}^{n}\Box\gamma_{k}))\models_{\Box Y}\pi^{'}$. As $((\vee_{i=1}^{l}\alpha_i)\vee(\vee_{k=1}^{n}\Box\gamma_{k}))$ is a theory implicate of $X$ with respect to $\Box Y$, so $((\vee_{i=1}^{l}\alpha_i)\vee(\vee_{k=1}^{n}\Box\gamma_{k}))\equiv_{\Box Y}\pi^{'}$. As $X\wedge\Box Y\models(\vee_{i=1}^{l}\alpha_i)\vee(\vee_{k=1}^{n}\Box\gamma_{k})$, so $(\vee_{i=1}^{l}\alpha_i)\vee(\vee_{k=1}^{n}\Box\gamma_{k})$ is an implicate of $X\wedge\Box Y$.  As there doesn't exist any $C$ such that $X\wedge\Box Y\models C$ and $C\models(\vee_{i=1}^{l}\alpha_i)\vee(\vee_{k=1}^{n}\Box\gamma_{k})$, so $(\vee_{i=1}^{l}\alpha_i)\vee(\vee_{k=1}^{n}\Box\gamma_{k})$ is a prime implicate of $X\wedge\Box Y$. Assuming $\pi=(\vee_{i=1}^{l}\alpha_i)\vee(\vee_{k=1}^{n}\Box\gamma_{k})$  we have for every $\pi^{'}\in\Theta(X, \Box Y)$ there exists a $\pi\in\Pi(X\cup\Box Y)$ such that $\pi^{'}\equiv_{\Box Y}\pi$.
\item $((\vee_{j=1}^{m}\beta_j)\vee(\vee_{k=1}^{n}\gamma_{k}))\wedge\neg\phi_1^{'}\wedge\ldots\wedge\neg\phi_q^{'}\models_{Y}\bot$. This implies, $((\vee_{j=1}^{m}\beta_j)\vee(\vee_{k=1}^{n}\gamma_{k}))\models_{Y}\phi_1^{'}\vee\ldots\vee\phi_q^{'}$. By Lemma \ref{symmetry_entailment}, we have $\Diamond((\vee_{j=1}^{m}\beta_j)\vee(\vee_{k=1}^{n}\gamma_{k}))\models_{\Box Y}\Diamond(\phi_1^{'}\vee\ldots\vee\phi_q^{'})$. By property of $\mathcal{K}$, $((\vee_{j=1}^{m}\Diamond\beta_j)\vee(\vee_{k=1}^{n}\Diamond\gamma_{k}))\models_{\Box Y}(\Diamond\phi_1^{'}\vee\ldots\vee\Diamond\phi_q^{'})$. Again by property of $\mathcal{T}$, $((\vee_{j=1}^{m}\Diamond\beta_j)\vee(\vee_{k=1}^{n}\Box\gamma_{k}))\models_{\Box Y}((\vee_{j=1}^{m}\Diamond\beta_j)\vee(\vee_{k=1}^{n}\Diamond\gamma_{k}))\models_{\Box Y}(\Diamond\phi_1^{'}\vee\ldots\vee\Diamond\phi_q^{'})\models_{\Box Y}\pi^{'}$. So $((\vee_{j=1}^{m}\Diamond\beta_j)\vee(\vee_{k=1}^{n}\Box\gamma_{k}))\models_{\Box Y}\pi^{'}$. As $((\vee_{j=1}^{m}\Diamond\beta_j)\vee(\vee_{k=1}^{n}\Box\gamma_{k}))$ is a theory implicate of $X$ with respect to $\Box Y$, so $((\vee_{j=1}^{m}\Diamond\beta_j)\vee(\vee_{k=1}^{n}\Box\gamma_{k}))\equiv_{\Box Y}\pi^{'}$. As $X\wedge\Box Y\models(\vee_{j=1}^{m}\Diamond\beta_j)\vee(\vee_{k=1}^{n}\Box\gamma_{k})$, so $(\vee_{j=1}^{m}\Diamond\beta_j)\vee(\vee_{k=1}^{n}\Box\gamma_{k})$ is an implicate of $X\wedge\Box Y$.  As there doesn't exist any $C$ such that $X\wedge\Box Y\models C$ and $C\models(\vee_{j=1}^{m}\Diamond\beta_j)\vee(\vee_{k=1}^{n}\Box\gamma_{k})$, so $(\vee_{j=1}^{m}\Diamond\beta_j)\vee(\vee_{k=1}^{n}\Box\gamma_{k})$ is a prime implicate of $X\wedge\Box Y$. Assuming $\pi=(\vee_{j=1}^{m}\Diamond\beta_j)\vee(\vee_{k=1}^{n}\Box\gamma_{k})$  we have for every $\pi^{'}\in\Theta(X, \Box Y)$ there exists a $\pi\in\Pi(X\cup\Box Y)$ such that $\pi^{'}\equiv_{\Box Y}\pi$.
\item $((\vee_{i=1}^{l}\alpha_i)\vee(\vee_{j=1}^{m}\beta_j)\vee(\vee_{k=1}^{n}\gamma_{k}))\wedge\neg\phi_1^{'}\wedge\ldots\wedge\neg\phi_q^{'}\models_{Y}\bot$. This implies, $((\vee_{i=1}^{l}\alpha_i)\vee(\vee_{j=1}^{m}\beta_j)\vee(\vee_{k=1}^{n}\gamma_{k}))\models_{Y}\phi_1^{'}\vee\ldots\vee\phi_q^{'}$. By Lemma \ref{symmetry_entailment}, $\Diamond((\vee_{i=1}^{l}\alpha_i)\vee(\vee_{j=1}^{m}\beta_j)\vee(\vee_{k=1}^{n}\gamma_{k}))\models_{\Box Y}\Diamond(\phi_1^{'}\vee\ldots\vee\phi_q^{'})$. By property of $\mathcal{K}$, $((\vee_{i=1}^{l}\Diamond\alpha_i)\vee(\vee_{j=1}^{m}\Diamond\beta_j)\vee(\vee_{k=1}^{n}\Diamond\gamma_{k}))\models_{\Box Y}(\Diamond\phi_1^{'}\vee\ldots\vee\Diamond\phi_q^{'})$. Again by property of $\mathcal{T}$, $((\vee_{i=1}^{l}\alpha_i)\vee(\vee_{j=1}^{m}\Diamond\beta_j)\vee(\vee_{k=1}^{n}\Box\gamma_{k}))\models_{\Box Y}((\vee_{i=1}^{l}\Diamond\alpha_i)\vee(\vee_{j=1}^{m}\Diamond\beta_j)\vee(\vee_{k=1}^{n}\Diamond\gamma_{k}))\models_{\Box Y}(\Diamond\phi_1^{'}\vee\ldots\vee\Diamond\phi_q^{'})\models_{\Box Y}\pi^{'}$. So, $((\vee_{i=1}^{l}\alpha_i)\vee(\vee_{j=1}^{m}\Diamond\beta_j)\vee(\vee_{k=1}^{n}\Box\gamma_{k}))\models_{\Box Y}\pi^{'}$. As $((\vee_{i=1}^{l}\alpha_i)\vee(\vee_{j=1}^{m}\Diamond\beta_j)\vee(\vee_{k=1}^{n}\Box\gamma_{k}))$ is a theory implicate of $X$ with respect to $\Box Y$ so, $((\vee_{i=1}^{l}\alpha_i)\vee(\vee_{j=1}^{m}\Diamond\beta_j)\vee(\vee_{k=1}^{n}\Box\gamma_{k}))\equiv_{\Box Y}\pi^{'}$. As $X\wedge\Box Y\models(\vee_{i=1}^{l}\alpha_i)\vee(\vee_{j=1}^{m}\Diamond\beta_j)\vee(\vee_{k=1}^{n}\Box\gamma_{k})$, so $(\vee_{i=1}^{l}\alpha_i)\vee(\vee_{j=1}^{m}\Diamond\beta_j)\vee(\vee_{k=1}^{n}\Box\gamma_{k})$ is an implicate of $X\wedge\Box Y$.  As there doesn't exist any $C$ such that $X\wedge\Box Y\models C$ and $C\models(\vee_{i=1}^{l}\alpha_i)\vee(\vee_{j=1}^{m}\Diamond\beta_j)\vee(\vee_{k=1}^{n}\Box\gamma_{k})$, so $(\vee_{i=1}^{l}\alpha_i)\vee(\vee_{j=1}^{m}\Diamond\beta_j)\vee(\vee_{k=1}^{n}\Box\gamma_{k})$ is a prime implicate of $X\wedge\Box Y$. Assuming $\pi=(\vee_{i=1}^{l}\alpha_i)\vee(\vee_{j=1}^{m}\Diamond\beta_j)\vee(\vee_{k=1}^{n}\Box\gamma_{k})$  we have for every $\pi^{'}\in\Theta(X, \Box Y)$ there exists a $\pi\in\Pi(X\cup\Box Y)$ such that $\pi^{'}\equiv_{\Box Y}\pi$.  Hence proved. \hfill{$\Box$}
\end{enumerate}
\end{proof}

\begin{remark}\label{tpi_pi}
Like Theorem \ref{tpi_size_less_than_pi} we can also prove that for every $\pi^{'}\in\Theta(X, \Box Y)$ there exists a $\pi\in\Pi(X\cup Y)$ such that $\pi^{'}\equiv_{\Box Y}\pi$. Consequently, $\mbox{nb\_cl}(\Theta(X,\Box Y)) \leq\mbox{nb\_cl}(\Pi(X\cup Y))$.
\end{remark}

\subsection{Algorithm for Computing Theory Prime Implicates}

Let us now present the algorithm for computation of theory prime implicates. The following algorithm is based on the Bienvenu's algorithm \cite{Bienvenu}. First, our algorithm computes the theory implicates of $X\cup\Box Y$ using the algorithm in \cite{Bienvenu} which is called as the set $CANDIDATES$. Here $Y$ is a propositional formula such that $X\models Y$. The assumption $X\models Y$ is considered in Definition \ref{omega} below. Then it computes theory implicates of $X\cup\Box Y$ and then it removes logically entailed clauses with respect to $\models_{\Box Y}$ to get theory prime implicates of $X\cup\Box Y$.

\vspace{.2cm}

\noindent{\bf Algorithm} \hspace{.2cm}  $MODALTPI(X,\Box Y)$ 

\vspace{.2cm}

\noindent  Input: Two formulas $X$ and $Y$ where $X$ is a modal formula and Y is any\\ 
$~~~~~~~~~~$ propositional formula\\
 Output: Set of theory prime implicates of $X$ with respect to $\Box Y$\\
{\tt begin\\
$~~~$ Compute CANDIDATES for $X\cup\Box Y$\\
$~~~$ Remove $\pi_j$ from CANDIDATES if $\pi_i\models_{\Box Y}\pi_j$ for some $\pi_i$ in CANDIDATES\\
$~~~$ Return CANDIDATES(=$\Theta(X,\Box Y))$\\
end}

\begin{theorem}\label{tpi_termination}
The algorithm MODALTPI terminates.
\end{theorem} 
\begin{proof}
To compute theory prime implicates of $X$ with respect to $\Box Y$, we are infact computing prime implicates of $X\cup\Box Y$. So by \cite{Bienvenu} the set CANDIDATES containing the set of implicates of  $X\cup\Box Y$ is finite. Then in CANDIDATES we compare a pair of implicates at most once for each pair and there are only finite such pairs, so the algorithm $MODALTPI$ must terminate. \hfill{$\Box$}
\end{proof}

The correctness of the above algorithm $MODALTPI$ follows from Theorem \ref{minimal_wrt_pi} and Theorem \ref{tpi_termination}.

\subsection{Theory Prime Implicate Compilation}

\begin{definition}\label{omega}
Let $X$ be a modal formula and $Y$ be any propositional formula such that $X\models Y$ and $\Box Y$ be tractable. The theory prime implicate compilation of $X$ with respect to $\Box Y$ is defined as $\Omega_{\Box Y}(X)=\Theta(X,\Box Y)\cup\Box Y$.
\end{definition}

\begin{theorem}\label{tpi_compilation}
$\Omega_{\Box Y}(X)\equiv X$.
\end{theorem}
\begin{proof}
We have to prove $\Theta(X,\Box Y)\cup\Box Y\equiv X$, i.e, to prove $\Theta(X,\Box Y)\cup\Box Y\models X$ and $X\models \Theta(X,\Box Y)\cup\Box Y$. First part is direct from Lemma \ref{equivalence_wrt_boxy}. For second part, as $X\models Y$ and $Y$ is a propositional formula, $Y\models\Box Y$, so $X\models\Box Y$, this implies, $X\models X\cup\Box Y$. So we can write it as $X\models X\cup\Box Y\cup\Box Y$. By Lemma \ref{equivalence_wrt_boxy}, $X\models\Theta(X,\Box Y)\cup\Box Y$. Hence proved. \hfill{$\Box$}
\end{proof}

The above result shows that $\Omega_{\Box Y}(X)$ is an equivalence preserving knowledge compilation and if you pose any query $Q$ to a knowledge base $X$ then it finds a propositional clause $Y$ contained in $X$ such that $X\models Y$ and compute the theory prime implicate of $X\cup\Box Y$ using the algorithm $MODALTPI$ and then the query $Q$ is answered from $\Omega_{\Box Y}(X)$ using the following algorithm $QA$ in polynomial time.

\vspace{.2cm}

\noindent{\bf Algorithm} \hspace{.2cm}  $QA(\Omega_{\Box Y}(X), Q)$

\vspace{.2cm}

\noindent  Input: The theory prime implicate compilation $\Omega_{\Box Y}(X)$ and a clausal query $Q$\\
 Output: true if $X\models Q$ holds\\
{\tt begin\\
$~~~$ if $\pi^{'}\models_{\Box Y} Q$ for every $\pi^{'}\in\Theta(X,\Box Y)\cup\Box Y$ \\
$~~~~~~~$ then return true\\
$~~~$ else \\
$~~~~~~~$ return false\\
end}

The correctness of the above algorithm follows from Theorem \ref{alt_def_of_tpi} and Theorem \ref{tpi_termination}. Let us now see how query answering can be performed in polynomial time.

\begin{theorem}\label{complexity}
Let $X$ be a modal formula and $Y$ be any propositional formula such that $X\models Y$ and $\Box Y$ is tractable. So checking whether $\pi^{'}\models_{\Box Y}Q$ holds in algorithm $QA$ can be done in time $O(|\Box Y\cup Q|^{m})$ and answering a query in algorithm $QA$ can be performed in time $O(|\Theta(X,\Box Y)\cup\Box Y|*|\Box Y\cup Q|^{m}))$.
\end{theorem}
\begin{proof}
 In order to check whether $\pi^{'}\models_{\Box Y}Q$ holds for each $\pi^{'}\in\Theta(X,\Box Y)\cup\Box Y$ in algorithm $QA$, we have to check whether $\Box Y\models\pi^{'}\rightarrow Q$ holds, i.e, to see whether $\Box Y\models\neg\pi^{'}\vee Q$ holds, i.e, to see whether $\Box Y\models\neg(l_1\vee l_2\vee\ldots\vee l_m)\vee Q$ holds where $l_{i}$'s are literals in $\pi^{'}$, i.e, to see whether $\Box Y\models(\neg l_1\wedge\neg l_2\wedge\ldots\wedge\neg l_m)\vee Q$ holds, i.e, to see whether $\Box Y\models(\neg l_1\vee Q)\wedge(\neg l_2\vee Q)\wedge\ldots\wedge(\neg l_m\vee Q)$ holds, i.e, to see whether $\Box Y\models(\neg l_i\vee Q)$ holds for each $i$ such that $1\leq i\leq m$. This test can be performed in time $|\Box Y\cup Q|$ for each $i$. So for all $i$ such that $1\leq i\leq m$ this test can be performed in time $O(|\Box Y\cup Q|^{m})$ where $\Box Y$ is tractable but this time complexity is for a single clause $\pi^{'}\in\Theta(X,\Box Y)\cup\Box Y$. So the query answering $QA$ for all the clauses can be performed in time $O(|\Theta(X,\Box Y)\cup\Box Y|*|\Box Y\cup Q|^{m}))$. \hfill{$\Box$}
\end{proof} 

\subsection{Tractable Theories}
The theory prime implicate compilation $\Theta(X,\Box Y)$ is easily exponential with respect to $|X\cup\Box Y|$ by \cite{Kean1} but if we have exponential number of queries then obviously each query can be answered in polynomial time in their size. As in query answering algorithm $QA$, we check whether  $\pi^{'}\models_{\Box Y} Q$ for every $\pi^{'}\in\Theta(X,\Box Y)\cup\Box Y$ and moreover query answering with respect to $\Theta(X,\Box Y)$ is polynomial so we assume $\Box Y$ to be tractable to keep the query answering in polynomial time. By \cite{Cerro}, the satisfiability of modal Horn clauses of $S5$ can be checked in polynomial time, so we assume $\Box Y$ to be a Horn clause in $S5$.

\begin{example}
 Consider a formula $X=(p_1\vee p_2)\wedge\Diamond\Box\neg p_3\wedge\Box\Diamond p_2$. We take $Y=X\setminus(\Diamond\Box\neg p_3\wedge\Box\Diamond p_2)=p_1\vee p_2$, so $\Box Y=\Box(p_1\vee p_2)$. Clearly $X\models Y$ and $\Box Y$ is tractable in system $S5$. So $X\wedge\Box Y=(p_1\vee p_2)\wedge\Diamond\Box\neg p_3\wedge\Box\Diamond p_2\wedge\Box(p_1\vee p_2)$. So we have $CANDIDATES=\{p_1\vee p_2, p_1\vee\Box(\Diamond p_2\wedge(p_1\vee p_2)), p_1\vee\Diamond(\Box\neg p_3\wedge\Diamond p_2\wedge(p_1\vee p_2)), \Box(\Diamond p_2\wedge(p_1\vee p_2))\vee p_2, \Box(\Diamond p_2\wedge(p_1\vee p_2)), \Box(\Diamond p_2\wedge(p_1\vee p_2))\vee\Diamond(\Box\neg p_3\wedge\Diamond p_2\wedge(p_1\vee p_2)), \Diamond(\Box\neg p_3\wedge\Diamond p_2\wedge(p_1\vee p_2))\vee p_2, \Diamond(\Box\neg p_3\wedge\Diamond p_2\wedge(p_1\vee p_2))\vee\Box(\Diamond p_2\wedge(p_1\vee p_2)), \Diamond(\Box\neg p_3\wedge\Diamond p_2\wedge(p_1\vee p_2))\}$. After removing logically entailed clauses with respect to $\Box Y$, the set of theory prime implicates of $X$ with respect to $\Box Y$, i.e, $\Theta(X,\Box Y)=\{p_1\vee p_2, \Box(\Diamond p_2\wedge(p_1\vee p_2)), \Diamond(\Box\neg p_3\wedge\Diamond p_2\wedge(p_1\vee p_2))\}$.
\end{example}

\section{Conclusion}
In this paper the definitions and results of theory prime implicates in propositional logic \cite{Marquis} is extended to modal logic $\mathcal{T}$ and the algorithm for computing theory prime implicates in propositional logic is also extended to modal logic according to \cite{Bienvenu} and its correctness has been proved. Another algorithm for query answering in \cite{Marquis} from $\Omega_{\Box Y}(X)$ is also extended to modal logic. Due to Lemma \ref{symmetry_entailment}, Theorem \ref{weaken_the_consequence}, and Theorem \ref{tpi_size_less_than_pi} we had to compute theory prime implicates of $X$ with respect to $\Box Y$ in the algorithm $MODALTPI$ instead of theory prime implicates of $X$ with respect to $Y$ as given in \cite{Marquis}. So, if $Y$ is empty, then $\Theta(X,\Box Y)=\Pi(X)$. As a future work, we want to compute theory prime implicates of a knowledge base $X$ with respect to another arbitrary modal knowledge base Z instead of the knowledge base $\Box Y$ for a proposition knowledge base $Y$ assumed here.  By Theorem \ref{tpi_compilation} we have shown that the theory prime implicate compilation $\Omega_{\Box Y}(X)$ is equivalent to $X$ so queries will be answered from $\Omega_{\Box Y}(X)$ in polynomial time by Theorem \ref{complexity}. Our algorithm $MODALTPI$ is based on Bienvenu's algorithm \cite{Bienvenu} which relies on distribution property whereas Marquis's theory prime implicate algorithm \cite{Marquis} is based on prime implicate generation algorithm of Kean and Tsiknis \cite{Kean1} and of de Kleer \cite{Kleer2} which rely on resolution.


%
%
%
%
%
\end{document}